\documentclass[11pt]{article} 
\pdfoutput=1

\usepackage{ifthen}
\def\acmconf{0} 

\ifthenelse{\equal{\acmconf}{1}}
{   

\usepackage{amsmath,amsthm,amssymb,xspace,verbatim,paralist,enumitem,multirow,bm,bbm}
\usepackage{array,multirow,booktabs}
\usepackage{multicol}
\usepackage{graphicx}
\usepackage{algorithm}
\usepackage[noend]{algpseudocode}
\usepackage{cleveref}
\usepackage{subcaption}

\ifthenelse{\equal{\acmconf}{1}}{%
\newcommand{\mypar}[1]{\medskip\noindent{\sffamily\bfseries #1.}~}
}{%
\usepackage{times,txfonts,microtype} 
\usepackage{caption}
\captionsetup[table]{font=small}
\captionsetup[figure]{font=small}
\newcommand{\mypar}[1]{\medskip\noindent{\bfseries #1.}~}
}

\usepackage{url}

\allowdisplaybreaks[1]

\newtheorem{theorem}{Theorem}[section]
\newtheorem{lemma}[theorem]{Lemma}
\newtheorem{corollary}[theorem]{Corollary}

\newtheorem{fact}[theorem]{Fact}

\theoremstyle{definition}  \newtheorem{definition}[theorem]{Definition}

\theoremstyle{remark}

\newcommand{\EE}{\mathbb{E}}

\newcommand{\NN}{\mathbb{N}}

\newcommand{\ZZ}{\mathbb{Z}}
\newcommand{\indic}{\mathbbm{1}}
\newcommand{\cA}{\mathcal{A}}

\newcommand{\cC}{\mathcal{C}}

\newcommand{\cG}{\mathcal{G}}

\newcommand{\bzero}{\mathbf{0}}
\newcommand{\be}{\mathbf{e}}
\newcommand{\bw}{\mathbf{w}}
\newcommand{\bx}{\mathbf{x}}
\newcommand{\by}{\mathbf{y}}
\newcommand{\bz}{\mathbf{z}}
\newcommand{\lbar}{\overline{\ell}}
\newcommand{\rbar}{\overline{r}}
\newcommand{\Lbar}{\overline{L}}
\newcommand{\Rbar}{\overline{R}}

\newcommand{\degen}{\kappa}

\newcommand{\etal}{{et al.}\xspace}

\newcommand{\eps}{\varepsilon}
\newcommand{\tO}{\widetilde{O}}

\newcommand{\ceil}[1]{\lceil{#1}\rceil}

\newcommand{\proc}[1]{\textsc{#1}}

\newcommand{\indx}{\textsc{index}\xspace}
\newcommand{\uor}{\textsc{unique-or}\xspace}
\newcommand{\nih}{\textsc{unique-find}\xspace}
\newcommand{\intfind}{\textsc{int-find}\xspace}

\newcommand{\gdct}{\textsc{graph-dist}\xspace}

\newcommand\Tstrut {\rule{0pt}{3ex}}         
\newcommand\Bstrut {\rule[-1.3ex]{0pt}{0pt}}   

\newcommand{\before}{\mathrel{\lhd}}
\renewcommand{\ge}{\geqslant}
\renewcommand{\le}{\leqslant}
\renewcommand{\geq}{\geqslant}
\renewcommand{\leq}{\leqslant}
\renewcommand{\b}{\{0,1\}}

\DeclareMathOperator{\Bin}{Bin}

\DeclareMathOperator{\Pair}{Pair}
\DeclareMathOperator{\Nbr}{Neighbor}

\DeclareMathOperator{\R}{R}
\DeclareMathOperator{\poly}{poly}
\DeclareMathOperator{\polylog}{polylog}
\DeclareMathOperator{\odeg}{odeg}
\DeclareMathOperator{\col}{\psi}

\newcommand{\Rdt}{\R^{\mathrm{dt}}}

\DeclareMathOperator{\local}{LOCAL}

\algdef{SE}[SUBALG]{Indent}{Unindent}{}{\algorithmicend\ }%
\algtext*{Indent}
\algtext*{Unindent}

    \setcopyright{rightsretained}

    \acmDOI{}

    \acmISBN{}

    \acmConference[PODS'19]{The Principles of Database Systems (PODS) symposium series}{July 2019}
    {Amsterdam, The Netherlands}
    \acmYear{2019}
    \copyrightyear{2019}

    \acmArticle{4}
    \acmPrice{15.00}

}
{
\usepackage[margin=1in]{geometry}
\usepackage[usenames,dvipsnames]{xcolor}
\usepackage{cite}
 
}

\begin{document}

\title{Graph Coloring via Degeneracy in Streaming and Other Space-Conscious Models}

\ifthenelse{\equal{\acmconf}{1}}
{

    \author{Suman K. Bera}
    \affiliation{%
        \institution{Dartmouth College}
        \city{Hanover}
        \state{NH}
        \postcode{03784}
    }
    \email{suman.k.bera.gr@dartmouth.edu}
    
    \author{Amit Chakrabarti}
    \affiliation{%
        \institution{Dartmouth College}
        \city{Hanover}
        \state{NH}
        \postcode{03784}
    }
    \email{amit.chakrabarti@dartmouth.edu}
    
    \author{Prantar Ghosh}
    \affiliation{%
        \institution{Dartmouth College}
        \city{Hanover}
        \state{NH}
        \postcode{03784}
    }
    \email{prantar.ghosh.gr@dartmouth.edu}
    \renewcommand{\shortauthors}{S. K. Bera \etal}
    
    %

\begin{abstract}
    
We study the problem of coloring a given graph using a small number of colors
in several well-established models of computation for big data. These include
the data streaming model, the general graph query model, the massively parallel
communication (MPC) model, and the CONGESTED-CLIQUE and the LOCAL models of 
distributed computation. On the one hand,
we give algorithms with sublinear complexity, for the appropriate notion of
complexity in each of these models. Our algorithms color a graph $G$ using
about $\degen(G)$ colors, where $\degen(G)$ is the degeneracy of $G$: this
parameter is closely related to the arboricity $\alpha(G)$. As a function of
$\degen(G)$ alone, our results are close to best possible, since the optimal
number of colors is $\degen(G)+1$.

On the other hand, we establish certain lower bounds indicating that
sublinear algorithms probably cannot go much further. In particular,
we prove that any randomized coloring algorithm that uses at most $\degen(G)+1$
colors would require $\Omega(n^2)$ storage in the one pass streaming model,
and $\Omega(n^2)$ many queries in the general graph query model, where $n$
is the number of vertices in the graph. These lower bounds hold even when the value
of $\degen(G)$ is known in advance; at the same time, our upper bounds do not require
$\degen(G)$ to be given in advance.

\end{abstract}




    \settopmatter{printacmref=false, printccs=true, printfolios=true}
    
}
{
    \author{%
        Suman K. Bera\thanks{%
                Department of Computer Science, Dartmouth College. 
                Supported in part by NSF under Award CCF-1650992.}%
        \and Amit Chakrabarti$^\fnsymbol{footnote}$
        \and Prantar Ghosh$^\fnsymbol{footnote}$%
    }
}

\maketitle

\ifthenelse{\equal{\acmconf}{1}}
{ }
{

\begin{abstract}
    
We study the problem of coloring a given graph using a small number of colors
in several well-established models of computation for big data. These include
the data streaming model, the general graph query model, the massively parallel
communication (MPC) model, and the CONGESTED-CLIQUE and the LOCAL models of 
distributed computation. On the one hand,
we give algorithms with sublinear complexity, for the appropriate notion of
complexity in each of these models. Our algorithms color a graph $G$ using
about $\degen(G)$ colors, where $\degen(G)$ is the degeneracy of $G$: this
parameter is closely related to the arboricity $\alpha(G)$. As a function of
$\degen(G)$ alone, our results are close to best possible, since the optimal
number of colors is $\degen(G)+1$.

On the other hand, we establish certain lower bounds indicating that
sublinear algorithms probably cannot go much further. In particular,
we prove that any randomized coloring algorithm that uses at most $\degen(G)+1$
colors would require $\Omega(n^2)$ storage in the one pass streaming model,
and $\Omega(n^2)$ many queries in the general graph query model, where $n$
is the number of vertices in the graph. These lower bounds hold even when the value
of $\degen(G)$ is known in advance; at the same time, our upper bounds do not require
$\degen(G)$ to be given in advance.

\end{abstract}

}

\section{Introduction}

Graph coloring is a fundamental topic in combinatorics and the corresponding
algorithmic problem of coloring an input graph with few colors is a basic and
heavily studied problem in computer science. 
It has numerous applications including in scheduling~\cite{thevenin2018graph,lotfi1986graph,leighton1979graph},
air traffic flow management~\cite{barnier2004graph}, frequency assignment
in wireless networks~\cite{balasundaram2006graph,park1996application}, register allocation~\cite{chaitin1982register,chow1990priority,chaitin1981register}. More recently, vertex coloring has been used to 
compute seed vertices in social networks that are
then expanded to detect community structures in the 
network~\cite{moradi2014local}.

Given an $n$-vertex graph $G =
(V,E)$, the task is to assign colors to the vertices in $V$ so that no two
adjacent vertices get the same color. Doing so with the minimum possible
number of colors---called the chromatic number, $\chi(G)$---is famously hard:
it is NP-hard to even approximate $\chi(G)$ to a factor of $n^{1-\eps}$ for
any constant $\eps >
0$~\cite{Feige1996zero,Zuckerman2006linear,Khot2006better}. In the face of
this hardness, it is algorithmically interesting to color $G$ with a
possibly suboptimal number of colors depending upon tractable parameters
of $G$. One such simple parameter is $\Delta$, the maximum degree: a trivial
greedy algorithm colors $G$ with $\Delta+1$ colors.

We study graph coloring in a number of space-constrained and
data-access-constrained settings, including the data streaming model and
certain distributed computing models. In such settings, finding a coloring
with ``about $\Delta$'' colors is a fairly nontrivial problem that has been
studied from various angles in a flurry of research over the last 
decade~\cite{AssadiCK19,Bhattacharya2018dynamic,Harris2016distributed,Chang2018optimal,ParterSu2018,Barenboim2016deterministic}.
In a recent breakthrough (awarded Best Paper at SODA~2019), Assadi, Chen, and Khanna~\cite{AssadiCK19}
gave sublinear algorithms for $(\Delta+1)$-coloring an input graph in such
models.

In this work, we focus on colorings that
use ``about $\degen$'' colors, where $\degen = \degen(G)$ is the {\em
degeneracy} of $G$, a parameter that improves upon $\Delta$. It is defined as
follows: $\degen = \min\{k:$ every induced subgraph of $G$ has a vertex of
degree at most $k\}$. Clearly, $\degen \le \Delta$. There is a simple greedy 
algorithm that runs in linear time and produces $(\degen + 1)$-coloring; see
\Cref{sec:prelim}. However, just as
before, when processing a massive graph under the constraints of either the
space-bounded streaming model or certain distributed computing models, the
inherently sequential nature of the greedy algorithm makes it infeasible.

\subsection{Our Results and Techniques} \label{sec:results}

We obtain a number of algorithmic results, as well as several lower bound results.

\mypar{Algorithms}
We give new graph coloring algorithms, parametrized by $\degen$, in the following models:
\begin{inparaenum}[\bfseries (1)]
  \item the data streaming model, where the input is a stream of edge insertions and deletions (i.e., a dynamic graph stream) resulting in the eventual graph to be colored and we are limited to a work space of $\tO(n)$ bits\footnote{The $\tO(\cdot)$ notation hides factors polylogarithmic in $n$.}, the so-called semi-streaming setting~\cite{FeigenbaumKMSZ05};
  \item the general graph query model~\cite{Goldreich-proptest-book}, where we may access the graph using only neighbor queries (what is the $i$th neighbor of $x$?) and pair queries (are $x$ and $y$ adjacent?);
  \item the massively parallel communication (MPC) model, where each of a large number of memory-limited processors holds a sublinear-sized portion of the input data and computation proceeds using rounds of communication;
  \item the congested clique model of distributed computation, where there is one processor per vertex holding that vertex's neighborhood information and each round allows each processor to communicate $O(\log n)$ bits to a specific other processor; and
  \item the LOCAL model of distributed computation, where there is one processor per vertex holding that vertex's neighborhood information and each round allows each processor to send an arbitrary amount of information to all its neighbors.
\end{inparaenum}

\begin{table*}[hbt]
\ifthenelse{\equal{\acmconf}{1}}%
  {}%
  {\small}%
\begin{minipage}{\textwidth} 
\centering
\begin{tabular}{c c rl c}
\toprule
{\bf Model}
& {\bf Number of Colors}  
& \multicolumn{2}{c } {\bf Complexity Parameters} 
& {\bf Source} \\

\midrule

{Streaming}
& {$\Delta+1$} 
& {\Tstrut $\tO(n)$} space,
& {$\tO(n\sqrt{\Delta})$} post-processing time
& \cite{AssadiCK19}
\\
{(one pass)}
& {\Tstrut $\degen + o(\degen)$ \Bstrut} 
& {$\tO(n)$} space, 
& {$\tO(n)$} post-processing time
& {this paper} \\
\hline

\multirow{2}{*}{ Query}
& {$\Delta+1$} 
&\multicolumn{2}{c} {\Tstrut $\tO(n^{3/2})$ queries }
& {\cite{AssadiCK19}}
\\
& {\Tstrut $\degen + o(\degen)$ \Bstrut} 
& \multicolumn{2}{c}{{$\tO(n^{3/2})$} queries}
& {this paper} \\
\hline

\multirow{2}{*}{MPC}
& {$\Delta+1$} 
& {$O(1)$} rounds,
& {\Tstrut $O(n \log^3 n)$} bits per processor
& {\cite{AssadiCK19}}
\\
& {\Tstrut $\degen + o(\degen)$ \Bstrut} 
& {$O(1)$} rounds,
& {$O(n \log^2 n)$} bits per processor
& {this paper} \\
\hline

\multirow{2}{*}{Congested Clique}
& {\Tstrut $\Delta+1$} 
& \multicolumn{2}{c}{$O(1)$ rounds}
& {\cite{CFGUZ2018}}
\\
& { \Tstrut $\degen + o(\degen)^{\star}$ \Bstrut} 
& \multicolumn{2}{c}{$O(1)$ rounds}
& {this paper} \\
\hline

\multirow{2}{*}{\Tstrut $\local$}
& {$O(\alpha n^{1/k})$} 
& {\Tstrut $O(k)$ rounds,} 
& {for~~ $k \in \big[ \omega(\log \log n),\, O(\sqrt{\log n}) \big]$}
& {\cite{Kothapalli2011distributed}}
\\
& {\Tstrut $O(\alpha n^{1/k} \log n)$} 
& { $O(k)$ rounds,} 
& {for $k \in \big[ \omega(\sqrt{\log n}),\, o(\log n) \big]$ \Bstrut}
& {this paper} 
\\
\bottomrule
\end{tabular}

\end{minipage}
\caption{Summary of our algorithmic results and basic comparison with most related previous work. In the result marked ($^\star$), we require that $\degen = \omega( \log^2 n)$. In the first two results, the number of colors can be improved to $\min\{\Delta+1, \degen+o(\degen)\}$ by running our algorithm alongside that of~\cite{AssadiCK19}; in the streaming setting, this would require knowing $\Delta$ in advance.}
\label{table:results} 
\end{table*}

\Cref{table:results} summarizes our algorithmic results and provides, in each case, a basic comparison with the most related result from prior work; more details appear in \Cref{sec:related}. As we have noted, $\degen \le \Delta$ in every case; indeed, $\degen$ could be arbitrarily better than $\Delta$ as shown by the example of a star graph, where $\degen = 1$ whereas $\Delta = n-1$. From a practical standpoint, it is notable that in many real-world large graphs drawn from various application domains---such as social networks, web graphs, and biological networks---the parameter $\degen$ is often {\em significantly} smaller than $\Delta$. See~\Cref{table:degen-vs-delta} for some concrete numbers.
That said, $\degen + o(\degen)$ is mathematically incomparable with $\Delta + 1$.



\begin{table*}[!hbt]
\ifthenelse{\equal{\acmconf}{1}}%
  {}%
  {\small}%

\begin{minipage}{\textwidth} 
\centering
\begin{tabular}{c c c c c}
\hline
{\bf Graph Name}
& {\bf $|V|$}  
& {\bf $|E|$} 
& {\bf $\Delta$} 
& {\bf $\degen$} \\
\hline
{soc-friendster}
& \Tstrut{$66$M} 
& {$2$B}
& {$5$K}
& {$305$}
\\
{fb-uci-uni}
& {$59$M} 
& {$92$M}
& {$5$K}
& {$17$}
\\
{soc-livejournal}
& {$4$M} 
& {$28$M}
& {$3$K}
& {$214$}
\\
{soc-orkut}
& {$3$M} 
& {$106$M}
& {$27$K}
& {$231$} \Bstrut
\\
{web-baidu-baike}
& {$2$M} 
& {$18$M}
& {$98$K}
& {$83$}
\\
{web-hudong}
& {$2$M} 
& {$15$M}
& {$62$K}
& {$529$}
\\
{web-wikipedia2009}
& {$2$M} 
& {$5$M}
& {$3$K}
& {$67$}
\\
{web-google}
& {$916$K} 
& {$5$M}
& {$6$K}
& {$65$}\Bstrut
\\
{bio-mouse-gene}
& {$43$K} 
& {$14$M}
& {$8$K}
& {$1$K}
\\
{bio-human-gene1}
& {$22$K} 
& {$12$M}
& {$8$K}
& {$2$K}
\\
{bio-human-gene2}
& {$14$K} 
& {$9$M}
& {$7$K}
& {$2$K}
\\
{bio-WormNet-v3}
& {$16$K} 
& {$763$K}
& {$1$K}
& {$165$}\Bstrut
\\
\hline
\end{tabular}

\end{minipage}
\caption{%
Statistics of several large real-world graphs taken from the application domains
of social networks, web graphs, and biological networks, showing that
the degeneracy, $\degen$, is often significantly smaller than 
the maximum degree, $\Delta$.
Source: \protect\url{http://networkrepository.com}~\cite{nr-aaai15}.%
}
\label{table:degen-vs-delta} 
\end{table*}

The parameter $\degen$ is also closely related to the {\em arboricity} $\alpha = \alpha(G)$, defined as the minimum number of forests into which the edges of $G$ can be partitioned. It is an easy exercise to show that $\alpha \le \degen \le 2\alpha-1$.

Perhaps even more than these results, our key contribution is a conceptual idea and a corresponding technical lemma underlying all our algorithms. We show that every graph admits a ``small'' sized {\em low degeneracy partition} (LDP), which is a partition of its vertex set into ``few'' blocks such that the subgraph induced by each block has low degeneracy, roughly logarithmic in $n$. Moreover, such an LDP can be computed by a very simple and distributed randomized algorithm: for each vertex, choose a ``color'' independently and uniformly at random from a suitable-sized palette (this is not to be confused with the eventual graph coloring we seek; this random assignment is most probably not a proper coloring of the graph). The resulting color classes define the blocks of such a partition, with high probability. \Cref{thm:ldp}, the LDP Theorem, makes this precise.

Given an LDP, a generic graph coloring algorithm is to run the aforementioned minimum-degree-based greedy algorithm on each block, using distinct palettes for the distinct blocks. We obtain algorithms achieving our claimed results by suitably implementing this generic algorithm in each computational model.

\mypar{Lower Bounds}
Recall that a graph with degeneracy $\degen$ admits a proper $(\degen+1)$-coloring. As \Cref{table:results} makes clear, there are several space-conscious $(\Delta+1)$-coloring algorithms known; perhaps we could aim for improved algorithms that provide $(\degen+1)$-colorings? We prove that this is not possible in sublinear space in either the streaming or the query model. In fact, our lower bounds prove more. We show that distinguishing $n$-vertex graphs of degeneracy $\degen$ from those with chromatic number $\degen+2$ requires $\Omega(n^2)$ space in the streaming model and $\Omega(n^2)$ queries in the general graph query model. This shows that it is hard to produce a $(\degen+1)$-coloring and in fact even to determine the value of $\degen$. These results generalize to the problems of producing a $(\degen+\lambda)$-coloring or estimating the degeneracy up to $\pm\lambda$; the corresponding lower bounds are $\Omega(n^2/\lambda^2)$. Furthermore, the streaming lower bounds hold even in the insertion-only model, where the input stream is simply a listing of the graph's edges in some order; compare this with our upper bound, which holds even for dynamic graph streams.

A possible criticism of the above lower bounds for coloring is that they seem to depend on it being hard to {\em estimate} the degeneracy $\degen$. Perhaps the coloring problem could become easier if $\degen$ was given to the algorithm in advance? We prove two more lower bounds showing that this is not so: the same $\Omega(n^2/\lambda^2)$ bounds hold even with $\degen$ known {\em a priori}.

Most of our streaming lower bounds use reductions from the \indx problem in communication complexity (a standard technique), via a novel gadget that we develop here; one bound uses a reduction from a variant of \textsc{disjointness}. Our query lower bounds use a related gadget and reductions from basic problems in Boolean decision tree complexity.

We conclude the paper with a ``combinatorial'' lower bound that addresses a potential criticism of our main algorithmic technique: the LDP. Perhaps a more sophisticated graph-theoretic result, such as the Palette Sparsification Theorem of Assadi et al.~(see below), could improve the quality of the colorings obtained? We prove that this is not so: there is no analogous theorem for colorings with ``about $\degen$'' colors.


\subsection{Related Work and Comparisons} \label{sec:related}

\mypar{Streaming and Query Models}
The work closest to ours in spirit is the recent breakthrough of Assadi, Chen, and Khanna~\cite{AssadiCK19}: they give a one-pass streaming $(\Delta+1)$-coloring algorithm that uses $\tO(n)$ space (i.e., is semi-streaming) and works on dynamic graph streams. Their algorithm exploits a key structural result that they establish: choosing a random $O(\log n)$-sized palette from $\{1,\ldots,\Delta+1\}$ for each vertex allows a compatible list coloring. They call this the {\em Palette Sparsification Theorem}. Their algorithm processes each stream update quickly, but then spends $\tO(n \sqrt{\Delta})$ time in post-processing. Our algorithm is similarly quick with the stream updates and has a faster $\tO(n)$-time post-processing step. Further, our algorithm is ``truly one-pass'' in that it does not require foreknowledge of $\degen$ or any other parameter of the input graph, whereas the Assadi et al.~algorithm needs to know the precise value of $\Delta$ before seeing the stream.

In the same paper, Assadi et al.~also consider the graph coloring problem in the query model. They give a $(\Delta+1)$-coloring algorithm that makes $\tO(n^{3/2})$ queries, followed by a fairly elaborate computation that runs in $\tO(n^{3/2})$ time and space. Our algorithm has the same complexity parameters and is arguably much simpler: its post-processing is just the straightforward greedy offline algorithm for $(\degen+1)$-coloring.

Another recent work on coloring in the streaming model is Radhakrishnan et al.~\cite{Radhakrishnan2015hypergraph}, which studies the problem of $2$-coloring an $n$-uniform hypergraph. In the query model, there are a number of works studying basic graph problems~\cite{Goldreich2008approximating,Parnas2007approximating,Chazelle2005approximating} but, to the best of our knowledge, Assadi et al.~were the first to study graph coloring in this sense. Also, to the best of our knowledge, there was no previously known algorithm for $O(\alpha)$-coloring in a semi-streaming setting, whereas here we obtain $(\degen + o(\degen))$-colorings; recall the bound $\degen \le 2\alpha - 1$.

\mypar{MPC and Congested Clique Models}
The MapReduce framework~\cite{DeanG04} is extensively used in distributed
computing to analyze and process massive data sets. Beame, Koutris, and Suciu~\cite{BeameKS13} defined the
Massively Parallel Communication (MPC) model to abstract out key theoretical features of MapReduce; it has since become a widely used setting for designing and analyzing big data algorithms, especially for graph problems. In this model, an input of size $m$ is distributed among $p \approx m/S$ processors, each of which is computationally unbounded and restricted to $S$ bits of space. The processors operate in synchronous rounds; in each round, a processor may communicate with all others, subject to the space constraint. The focus is on using a very small number of rounds.

Another well-studied model for distributed graph algorithms is Congested Clique~\cite{LotkerPPP05}, where there are $n$ nodes, each holding the local neighborhood information for one of the $n$ vertices of the input graph. The nodes communicate in synchronous rounds; in a round, every pair of processors may communicate, but each message is restricted to $O(\log n)$ bits. Behnezhad et al.~\cite{Behnezhad2018BriefAS} show that Congested Clique is equivalent to the so-called ``semi-MPC model,'' defined as MPC with $O(n\log n)$ bits of memory per machine: there are simulations in both directions preserving the round complexity.

Graph coloring has been studied in these models before. Harvey et al.~\cite{HarveyLL18} gave a $(\Delta+o(\Delta))$-coloring algorithm in
the MapReduce model; it can be simulated in MPC using $O(1)$ rounds and $O(n^{1+c})$ space per machine for some constant $c>0$. Parter~\cite{Parter18} gave a Congested Clique algorithm for $(\Delta + 1)$-coloring using $O(\log \log \Delta \cdot \log^{\star}\Delta)$ rounds; Parter and Su~\cite{ParterSu2018} improved this
to $O(\log^{\star} \Delta)$. The aforementioned paper of Assadi et al.~\cite{AssadiCK19} gives an MPC algorithm for $(\Delta+1)$-coloring using $O(1)$-round and $O(n\log^3 n)$ bits of space per machine. Because this space usage is $\omega(n\log n)$, the equivalence result of Behnezad et al.~\cite{Behnezhad2018BriefAS} does not apply and this doesn't lead to an $O(1)$-round Congested Clique algorithm. In contrast, our MPC algorithm uses only $O(n\log n)$ bits of space per machine for graphs with degeneracy $\omega(\log^2 n)$, and therefore leads to such a Congested Clique algorithm. Chang \etal~\cite{CFGUZ2018} have recently designed two $(\Delta + 1)$ list-coloring
algorithms: an $O(1)$-round Congested Clique algorithm, and an 
$O(\sqrt{\log \log n})$-round MPC algorithm with $o(n)$ space per machine and $\tO(m)$ space in total. To the best of our knowledge, no
$O(\alpha)$-coloring algorithm was previously known, in either the MPC or the Congested Clique model.

\mypar{The LOCAL Model}
The LOCAL model of distributed computing is ``orthogonal'' to Congested Clique: the input setup is similar but, during computation, each node may only communicate with its neighbors in the input graph, though it may send an arbitrarily long message. As before, the focus is on minimizing the number of rounds (a.k.a., time). There is a deep body of work on graph coloring in this model. Indeed, graph coloring is one of {\em the} most central ``symmetry breaking'' problems in distributed computing. We refer the reader to the monograph by Barenboim and Elkin~\cite{BarenboimElkin-book} for an excellent overview of the state of the art. Here, we shall briefly discuss only a few results closely related to our contribution.

There is a long line of work on fast $(\Delta+1)$-coloring in the LOCAL model, in the deterministic as well as the randomized setting~\cite{Panconesi1996complexity,Barenboim2016deterministic,Fraigniaud2016local,Luby1986simple,Johansson1999simple,Alon1986fast,Schneider2010new,Barenboim2016locality} culminating in sublogarithmic time solutions due to Harris~\cite{Harris2016distributed} and Chang \etal~\cite{Chang2018optimal}.
%
%
%
Barenboim and Elkin \cite{Barenboim2010sublogarithmic,Barenboim2011deterministic} 
studied fast distributed coloring algorithms that may use far fewer than $\Delta$ colors: in particular, they gave algorithms that use $O(\alpha)$ colors and run in $O(\alpha^\eps \log n)$ time on graphs with arboricity at most $\alpha$. Recall again that
$\degen \le 2\alpha-1$, so that a $2\alpha$-coloring always {\em exists}. They also gave a faster $O(\log n)$-time algorithm using $O(\alpha^2)$ colors. Further, they gave a family of algorithms that produce an $O(t\alpha^2)$-coloring in $O(\log_t n + \log^\star n)$, for every $t$ such that $2\leq t \leq O(\sqrt{{n}/{\alpha}})$. Our algorithm for the LOCAL model builds on this latter result.


Kothapalli and Pemmaraju~\cite{Kothapalli2011distributed} focused on 
arboricity-dependant coloring using very few rounds. They gave a randomized $O(k)$-round algorithm that uses $O(\alpha n^{1/k})$ colors for $2\log \log n \leq k \leq \sqrt{\log n}$ and
$O(\alpha^{1+1/k} n^{1/k+3/k^2} 2^{-2^k})$ colors for $k < 2\log \log n$.
We extend their result to the range
$k \in \big[ \omega(\sqrt{\log n}),\, o(\log n) \big]$, using $O(\alpha n^{1/k}\log n)$ colors.

Ghaffari and Lymouri~\cite{Ghaffari2017simple}
gave a randomized
$O(\alpha)$-coloring algorithm that runs in time
$O(\log n \cdot \min \{\log \log n, \log \alpha \})$ as well as an $O(\log n)$-time algorithm
using 
$\min \{ (2+\eps)\alpha + O(\log n \log\log n),$ $O(\alpha \log \alpha)\}$ colors,
for any constant $\eps > 0$. However, their technique does not yield a sublogarithmic time algorithm, even at the cost of a larger palette.

\mypar{The LDP Technique}
As mentioned earlier, our algorithmic results rely on the concept of a low degeneracy partition (LDP) that we introduce in this work. Some relatives of this idea have been considered before. Specifically, Barenboim and Elkin~\cite{BarenboimElkin-book} define a $d$-defective (resp.~$b$-arbdefective) $c$-coloring to be a vertex coloring using palette $[c]$ such that every color class induces a subgraph with maximum degree at most $d$ (resp.~arboricity at most $b$). Obtaining such improper colorings is a useful first step towards obtaining proper colorings. They give deterministic algorithms to obtain good arbdefective colorings~\cite{Barenboim2011deterministic}. However, their algorithms are elaborate and are based on construction of low outdegree acyclic partial orientations of the graph's edges: an expensive step in our space-conscious models.

Elsewhere (Theorem~10.5 of Barenboim and Elkin~\cite{BarenboimElkin-book}), they note that a useful defective (not arbdefective) coloring is easily obtained by randomly picking a color for each vertex; this is then useful for computing an $O(\Delta)$-coloring. 

Our LDP technique can be seen as a simple randomized method for producing an arbdefective coloring. Crucially, we parametrize our result using degeneracy instead of arboricity and we give sharp---not just asymptotic---bounds on the degeneracy of each color class. 

\mypar{Other Related Work}
Other work considers coloring in the setting of {\em dynamic graph algorithms}: edges are inserted and deleted over time and the goal is to 
{\em maintain} a valid vertex coloring of the graph that must be updated quickly after each modification. Unlike in the streaming setting,
there is no space restriction. Bhattacharya \etal~\cite{Bhattacharya2018dynamic}
gave a randomized algorithm that maintains a $(\Delta+1)$-coloring with $O(\log \Delta)$ expected amortized update time and a deterministic algorithm that maintains a $(\Delta + o(\Delta))$-coloring with $O(\polylog \Delta)$ amortized update time. Barba \etal~\cite{barba2017dynamic}
gave tradeoffs between the number of colors used and
update time. However, the techniques in these works do not seem to apply in the streaming setting due to fundamental differences in the models.

Estimating the arboricity of a graph in the streaming model
is a well studied problem. McGregor \etal~\cite{Mcgregor2015densest}
gave a one pass $(1+\eps)$-approximation algorithm to 
estimate the arboricity of graph using $\tO(n)$ space. 
Bahmani \etal~\cite{Bahmani2012densest} gave a matching
lower bound. Our lower bounds for estimating degeneracy are quantitatively much larger but they call for much tighter estimates.


\section{Preliminaries} \label{sec:prelim}

Throughout this paper, graphs are simple, undirected, and unweighted. In considering a graph coloring problem, the input graph will usually be called $G$ and we will put $n = |V(G)|$. The notation ``$\log x$'' stands for $\log_2 x$. For an integer $k$, we denote the set $\{1,2,\ldots,k\}$ by $[k]$. 

For a graph $G$, we define $\Delta(G) = \max\{\deg(v):\, v \in V(G)\}$. We say that $G$ is $k$-degenerate if every induced subgraph of $G$ has a vertex of degree at most $k$. For instance, every forest is $1$-degenerate and an elementary theorem says that every planar graph is $5$-degenerate. The {\em degeneracy} $\degen(G)$ is the smallest $k$ such that $G$ is $k$-degenerate. The {\em arboricity} $\alpha(G)$ is the smallest $r$ such that the edge set $E(G)$ can be partitioned into $r$ forests. When the graph $G$ is clear from the context, we simply write $\Delta$, $\degen$, and $\alpha$, instead of $\Delta(G)$, $\degen(G)$, and $\alpha(G)$.

We note two useful facts: the first is immediate from the definition, and the second is an easy exercise.

\begin{fact} \label{fact:degen-sparse}
  If an $n$-vertex graph has degeneracy $\degen$, then it has at most $\degen n$ edges. \qed
\end{fact}

\begin{fact} \label{fact:degen-arb}
  In every graph, the degeneracy $\degen$ and arboricity $\alpha$ satisfy $\alpha \le \degen \le 2\alpha-1$. \qed
\end{fact}

In analyzing our algorithms, it will be useful to consider certain {\em vertex orderings} of graphs and their connection with the notion of degeneracy, given by \Cref{lem:degen-odeg} below. Although the lemma is folklore, it is crucial to our analysis, so we include a proof for completeness.
\begin{definition} \label{def:odeg}
  An {\em ordering} of $G$ is a list consisting of all its vertices (equivalently, a total order on $V(G)$). Given an ordering $\before$, for each $v \in V(G)$, the {\em ordered neighborhood}
  \ifthenelse{\equal{\acmconf}{1}}
  {$N_{G,\before}(v) := \{w \in V(G):\, \{v,w\} \in E(G), v \before w\}$,}
  {\[ N_{G,\before}(v) := \{w \in V(G):\, \{v,w\} \in E(G), v \before w\} \,, \]}
  i.e., the set of neighbors of $v$ that appear {\em after} $v$ in the ordering. The {\em ordered degree} $\odeg_{G,\before}(v) := |N_{G,\before}(v)|$.  
\end{definition}

\begin{definition} \label{def:degen-order}
  A {\em degeneracy ordering} of $G$ is an ordering produced by the following algorithm: starting with an empty list, repeatedly pick a minimum degree vertex $v$ (breaking ties arbitrarily), append $v$ to the end of the list, and delete $v$ from $G$; continue this until $G$ becomes empty.
\end{definition}

\begin{lemma} \label{lem:degen-odeg}
  A graph $G$ is $k$-degenerate iff there exists an ordering $\before$ such that $\odeg_{G,\before}(v) \le k$ for all $v \in V(G)$.
\end{lemma}
\begin{proof}
  Suppose that $G$ is $k$-degenerate. Let $\before\, = (v_1, \ldots, v_n)$ be a degeneracy ordering. Then, for each $i$, $\odeg_{G,\before}(v_i)$ is the degree of $v_i$ in the induced subgraph $G \setminus \{v_1, \ldots, v_{i-1}\}$. By definition, this induced subgraph has a vertex of degree at most $k$, so $v_i$, being a minimum degree vertex in the subgraph, must have degree at most $k$.
  
  On the other hand, suppose that $G$ has an ordering $\before$ such that $\odeg_{G,\before}(v) \le k$ for all $v \in V(G)$. Let $H$ be an induced subgraph of $G$. Let $v$ be the leftmost (i.e., smallest) vertex in $V(H)$ according to $\before$. Then all neighbors of $v$ in $H$ in fact lie in $N_{G,\before}(v)$, so $\deg_H(v) \le \odeg_{G,\before}(v) \le k$. Therefore, $G$ is $k$-degenerate.
\end{proof}

A $c$-coloring of a graph $G$ is a mapping $\col \colon V(G)\to [c]$; it is said to be a {\em proper coloring} if it makes no edge monochromatic: $\col(u) \ne \col(v)$ for all $\{u,v\} \in E(G)$. The smallest $c$ such that $G$ has a proper $c$-coloring is called the {\em chromatic number} $\chi(G)$. By considering the vertices of $G$ one at a time and coloring greedily, we immediately obtain a proper $(\Delta+1)$-coloring. This idea easily extends to degeneracy-based coloring.

\begin{lemma} \label{lem:degen-color}
  Given unrestricted (``offline'') access to an input graph $G$, we can produce a proper $(\degen+1)$-coloring in linear time.
\end{lemma}
\begin{proof}
  Construct a degeneracy ordering $(v_1, \ldots, v_n)$ of $G$ and then consider the vertices one by one in the order $(v_n, \ldots, v_1)$, coloring greedily. Given a palette of size $\degen+1$, by the ``only if'' direction of \Cref{lem:degen-odeg}, there will always be a free color for a vertex when it is considered.
\end{proof}

Of course, the simple algorithm above is not implementable directly in ``sublinear'' settings, such as space-bounded streaming algorithms, query models, or distributed computing models. Nevertheless, we shall make use of the algorithm on suitably constructed subgraphs of our input graph.

We shall use the following form of the Chernoff bound.
\begin{fact} \label{fact:chernoff}
  Let $X$ be a sum of mutually independent indicator random variables. Let $\mu$ and $\delta$ be real numbers such that $\EE X \le \mu$ and $0 \le \delta \le 1$. Then,
  $\Pr\left[ X \ge (1+\delta)\mu \right] \le \exp\left( - \mu \delta^2 / 3 \right)$. \qed
\end{fact}


\section{A Generic Framework for Coloring} \label{sec:framework}

In this section, we give a generic framework for graph coloring that we later instantiate in various computational models. As a reminder, our focus is on graphs $G$ with a nontrivial upper bound on the degeneracy $\degen = \degen(G)$. Each such graph {\em admits} a proper $(\degen+1)$-coloring; our focus will be on obtaining a proper $(\degen+o(\degen))$-coloring efficiently.

As a broad outline, our framework calls for coloring $G$ in two phases. The first phase produces a {\em low degeneracy partition} (LDP) of $G$: it partitions $V(G)$ into a ``small'' number of parts, each of which induces a subgraph that has ``low'' degeneracy. This step can be thought of as preprocessing and it is essentially free (in terms of complexity) in each of our models. The second phase properly colors each part, using a small number of colors, which is possible because the degeneracy is low. In \Cref{sec:algs}, we shall see that  the low degeneracy allows this second phase to be efficient in each of the models we consider.


\subsection{A Low Degeneracy Partition and its Application} \label{sec:ldp}

In this phase of our coloring framework, we assign each vertex a color chosen uniformly at random from $[\ell]$, these choices being mutually independent, where $\ell$ is a suitable parameter. For each $i \in [\ell]$, let $G_i$ denote the subgraph of $G$ induced by vertices colored $i$. We shall call each $G_i$ a {\em block} of the vertex partition given by $(G_1, \ldots, G_\ell)$. The next theorem, our main technical tool, provides certain guarantees on this partition given a suitable choice of $\ell$.

\begin{theorem}[LDP Theorem]
\label{thm:ldp}
Let $G$ be an $n$-vertex graph with degeneracy $\degen$. Let $k \in [1,n]$ be a ``guess'' for the value of $\degen$ and let $s \ge Cn\log n$ be a sparsity parameter, where $C$ is a sufficiently large universal constant. Put
\begin{equation} \label{eq:ldef}
  \ell = \left\lceil \frac{2nk}{s} \right\rceil \,, \quad \lambda = 3\sqrt{\degen \ell \log n} \,,
\end{equation}
and let $\psi \colon V(G) \to [\ell]$ be a uniformly random coloring of $G$. For $i \in [\ell]$, let $G_i$ be the subgraph induced by $\psi^{-1}(i)$. Then, the partition $(G_1, \ldots, G_\ell)$ has the following properties.
\begin{enumerate}[label=(\roman*)]
    \item \label{ldp:degen} If $k \le 2\degen$, then w.h.p., for each $i$, the degeneracy $\degen(G_i) \le (\degen + \lambda)/\ell$.
    \item \label{ldp:block} W.h.p., for each $i$, the block size $|V(G_i)| \le 2n/\ell$.
    \item \label{ldp:sparse} If $\degen \le k \le 2\degen$, then w.h.p., the number of monochromatic edges $|E(G_1) \cup \cdots \cup E(G_\ell)| \le s$.
\end{enumerate}
In each case, ``w.h.p.'' means ``with probability at least $1-1/\poly(n)$.''
\end{theorem}

It will be convenient to encapsulate the guarantees of this theorem in a definition.

\begin{definition} \label{def:ldp}
  Suppose graph $G$ has degeneracy $\degen$. A vertex partition $(G_1, \ldots, G_\ell)$ simultaneously satisfying the degeneracy bound in \cref{ldp:degen}, the block size bound in \cref{ldp:block}, and the (monochromatic) edge sparsity bound in \cref{ldp:sparse} in \Cref{thm:ldp} is called an $(\ell,s,\lambda)$-LDP of $G$.
\end{definition}

It will turn out that an $(\ell,s,\lambda)$-LDP leads to a proper coloring of $G$ using at most $\degen+\lambda+\ell$ colors. An instructive setting of parameters is $s = \Theta((n\log n)/\eps^2)$, where $\eps$ is either a small constant or a slowly vanishing function of $n$, such as $1/\log n$. Then, a quick calculation shows that when an accurate guess $k \in [\degen, 2\degen]$ is made, \Cref{thm:ldp} guarantees an LDP that has edge sparsity $s = \tO(n)$ and that leads to an eventual proper coloring using $(1 + O(\eps))\degen$ colors. When $\eps = o(1)$, this number of colors is $\degen + o(\degen)$.

Recall that the second phase of our coloring framework involves coloring each $G_i$ separately, exploiting its low degeneracy. Indeed, given an $(\ell,s,\lambda)$-LDP, each block $G_i$ {\em admits} a proper $(\degen(G_i)+1)$-coloring. Suppose we use a distinct palette for each block; then the total number of colors used is
\begin{equation} \label{eq:numcolors}
  \sum_{i=1}^\ell (\degen(G_i) + 1) \le \ell \left( \frac{\degen + \lambda}{\ell} + 1 \right) = \degen+\lambda+\ell \,,
\end{equation}
as claimed above. Of course, even if our first phase random coloring $\psi$ yields a suitable LDP, we still have to collect each block $G_i$ or at least enough information about each block so as to produce a proper $(\degen(G_i)+1)$-coloring. How we do this depends on the precise model of computation. We take this up in \Cref{sec:algs}.

\subsection{Proof of the LDP Theorem} \label{sec:ldp-proof}

We now turn to proving the LDP Theorem from \Cref{sec:ldp}. Notice that when $k \le (C/2)\log n$, the condition $s \ge Cn\log n$ results in $\ell = 1$, so the vertex partition is the trivial one-block partition, which obviously satisfies all the properties in the theorem. Thus, in our proof, we may assume that $k > (C/2)\log n$.

\begin{proof}[Proof of \Cref{thm:ldp}]
  We start with \cref{ldp:block}, which is the most straightforward. From \cref{eq:ldef}, we have $\ell \le 4nk/s$, so
  \[
    \frac{n}{\ell} \ge \frac{s}{4k} \ge \frac{Cn\log n}{4k} \ge \frac{C\log n}{4} \,.
  \]
  Each block size $|V(G_i)|$ has binomial distribution $\Bin(n,1/\ell)$, so a Chernoff bound gives
  \[
    \Pr\left[ |V(G_i)| > \frac{2n}{\ell} \right] \le \exp\left( -\frac{n}{3\ell} \right) \le \exp\left( -\frac{C \log n}{12} \right) \le \frac{1}{n^2} \,,
  \]
  for sufficiently large $C$. By a union bound over the at most $n$ blocks, \cref{ldp:block} fails with probability at most $1/n$. \medskip
  
  \Cref{ldp:degen,ldp:sparse} include the condition $k \le 2\degen$, which we shall assume for the rest of the proof. By \cref{eq:ldef} and the bounds $s \ge Cn\log n$ and $k > (C/2)\log n$,
  \[
    \ell \le \left\lceil \frac{2k}{C\log n} \right\rceil
    \le \frac{4k}{C\log n}
    \le \frac{8\degen}{C\log n} \,,
  \]
  whence, for sufficiently large $C$,
  \begin{equation} \label{eq:lambda-degen}
    \lambda \le 3\sqrt{\degen \cdot \frac{8\degen}{C\log n} \cdot \log n} \le \degen \,.
  \end{equation}
  
  We now turn to establishing \cref{ldp:degen}. Let $\before$ be a degeneracy ordering for $G$. For each $i \in [\ell]$, let $\before_i$ be the restriction of $\before$ to $V(G_i)$. Consider a particular vertex $v \in V(G)$ and let $j = \psi(v)$ be its color. We shall prove that, w.h.p., $\odeg_{G,\before_j}(v) \le (\degen+\lambda)/\ell$.
  
  By the ``only if'' direction of \Cref{lem:degen-odeg}, we have $\odeg_{G,\before}(v) = |N_{G,\before}(v)| \le \degen$. Now note that
  \[
    \odeg_{G_j, \before_j}(v) = \sum_{u \in N_{G,\before(v)}} \indic_{\{\psi(u) = \psi(v)\}}
  \]
  is a sum of mutually independent indicator random variables, each of which has expectation $1/\ell$. Therefore, $\EE \odeg_{G_j, \before_j}(v) = \odeg_{G,\before}(v)/\ell \le \degen/\ell$. Since $\lambda \le \degen$ by \cref{eq:lambda-degen}, we may use the form of the Chernoff bound in \Cref{fact:chernoff}, which gives us
  \[
    \Pr\left[ \odeg_{G_j, \before_j}(v) > \frac{\degen+\lambda}{\ell} \right]
    \le \exp\left( - \frac{\degen}{\ell} \frac{\lambda^2}{3\degen^2} \right)
    = \exp\left( - \frac{9 \degen\ell\log n}{3\degen\ell} \right)
    \le \frac{1}{n^3} \,,
  \]
  where the equality follows from \cref{eq:ldef}. In words, with probability at least $1-1/n^3$, the vertex $v$ has ordered degree at most $(\degen+\lambda)/\ell$ within its own block. By a union bound, with probability at least $1-1/n^2$, all $n$ vertices of $G$ satisfy this property. When this happens, by the ``if'' direction of \Cref{lem:degen-odeg}, it follows that $\degen(G_i) \le (\degen+\lambda)/\ell$ for every $i$. \medskip
  
  Finally, we take up \cref{ldp:sparse}, which is now straightforward. Assume that the high probability event in \cref{ldp:degen} occurs. Then, by \Cref{fact:degen-sparse},
  \ifthenelse{\equal{\acmconf}{1}}
  { \begin{multline*} }
  { \[ }
    |E(G_1) \cup \cdots \cup E(G_\ell)|
    \le \sum_{i=1}^\ell \degen(G_i)\, |V(G_i)| 
    \le \frac{\degen+\lambda}{\ell} \sum_{i=1}^\ell |V(G_i)|
    = \frac{n(\degen + \lambda)}{\ell}
    \le \frac{2n\degen}{\ell}
    \le s \,,
  \ifthenelse{\equal{\acmconf}{1}}    
  { \end{multline*} }
  { \] }
  where the final inequality uses the condition $\degen \le k$ and \cref{eq:ldef}.
\end{proof}


\section{Specific Sublinear Algorithms for Coloring} \label{sec:algs}

We now turn to designing graph coloring algorithms in specific models of computation for big data, where the focus is on utilizing space sublinear in the size of the massive input graph. Such models are sometimes termed {\em space-conscious}. In each case, our algorithm ultimately relies on the framework developed in \Cref{sec:framework}.

\subsection{Data Streaming} \label{sec:streaming}

We begin with the most intensely studied space-conscious model: the data streaming model. For graph problems, in the basic model, the input is a stream of non-repeated edges that define the input graph $G$: this is called the {\em insertion-only} model, since it can be thought of as building up $G$ through a sequence of edge insertions. In the more general {\em dynamic graph model} or {\em turnstile model}, the stream is a sequence of edge updates, each update being either an insertion or a deletion: the net effect is to build up $G$. Our algorithm will work in this more general model. Later, we shall give a corresponding lower bound that will hold even in the insertion-only model (for a lower bound, this is a strength).

We assume that the vertex set $V(G)=[n]$ and the input is a stream $\sigma$ of at most $m = \poly(n)$ updates to an initially empty graph. An update is a triple $(u,v,c)$, where $u,v\in V(G)$ and $c \in \{-1,1\}$: when $c=1$, this token represents an insertion of edge $\{u,v\}$ and when $c=-1$, it represents a deletion. Let $N = \binom{n}{2}$ and $[[m]] = \ZZ \cap [-m,m]$. It is convenient to imagine a vector $\bx \in [[m]]^N$ of edge multiplicities that starts at zero and is updated entrywise with each token. The input graph $G$ described by the stream will be the underlying simple graph, i.e., $E(G)$ will be the set of all edges $\{u,v\}$ such that $x_{u,v} \ne 0$ at the end. We shall say that $\sigma$ {\em builds up} $\bx$ and $G$.

Our algorithm makes use of two data streaming primitives, each a {\em linear sketch}. (We can do away with these sketches in the insertion-only setting; see the end of this section.) The first is a sketch for {\em sparse recovery} given by a matrix $A$ (say): given a vector $\bx \in [[m]]^N$ with sparsity $\|\bx\|_0 \le t$, there is an efficient algorithm to reconstruct $\bx$ from $A\bx$. The second is a sketch for {\em $\ell_0$ estimation} given by a random matrix $B$ (say): given a vector $\bx \in [[m]]^N$, there is an efficient algorithm that takes $B\bx$ and computes from it an estimate of $\|\bx\|_0$ that, with probability at least $1-\delta$, is a $(1+\gamma)$-multiplicative approximation. It is known that there exists a suitable $A \in \b^{d\times N}$, where $d = O(t \log(N/t))$, where $A$ has column sparsity $O(\log(N/t))$; see, e.g., Theorem~9 of Gilbert and Indyk~\cite{GilbertI10}. It is also known that there exists a suitable distribution over matrices giving $B \in \b^{d'\times N}$ with $d' = O(\gamma^{-2} \log\delta^{-1} \log N (\log\gamma^{-1} + \log\log m))$. Further, given an update to the $i$th entry of $\bx$, the resulting updates in $A\bx$ and $B\bx$ can be effected quickly by generating the required portion of the $i$th columns of $A$ and $B$.

\algrenewcommand\algorithmicforall{\textbf{foreach}}
\begin{algorithm*}[!ht]
    \caption{One-Pass Streaming Algorithm for Graph Coloring via Degeneracy}
    \label{alg:streaming}
    \begin{algorithmic}[1]
        \Procedure{Color}{stream $\sigma$, integer $k$}
        \Comment{$\sigma$ builds up $\bx$ and $G$;~ $k \in [1,n]$ is a guess for $\degen(G)$}
            \State choose $s$, $\ell$ as in \cref{eq:ldef} and $t,d,d',A,B$ as in the above discussion
            \State initialize $\by \in [[m]]^d$ and $\bz \in [[m]]^{d'}$ to zero
            \ForAll{$u \in [n]$} $\psi(u) \gets$ uniform random color in $[\ell]$ \EndFor
            \ForAll{token $(u,v,c)$ in $\sigma$}
                \If{$\psi(u) = \psi(v)$} $\by \gets \by + cA_{u,v}$;~ $\bz \gets \bz + cB_{u,v}$ \EndIf
            \EndFor
            \If{estimate of $\|\bw\|_0$ obtained from $\bz$ is $> 5s/4$} \textbf{abort} \EndIf
            \State $\bw' \gets$ result of $t$-sparse recovery from $\by$
            \Comment{we expect that $\bw' = \bw$}
            \ForAll{$i \in [\ell]$}
                \State $G_i \gets$ simple graph induced by $\{\{u,v\}:\, w'_{u,v} \ne 0$ and $\psi(u)=\psi(v)=i\}$
                \State \textbf{color} $G_i$ using palette $\{(i,j):\, 1 \le j \le \degen(G_i)+1\}$; cf.~\Cref{lem:degen-color}
                \Comment{net effect is to color $G$}
            \EndFor
        \EndProcedure
    \end{algorithmic}
\end{algorithm*}

In our description of \Cref{alg:streaming}, we use $A_{u,v}$ (resp.~$B_{u,v}$) to denote the column of $A$ (resp.~$B$) indexed by $\{u,v\}$. The algorithm's logic results in sketches $\by = A\bw$ and $\bz = B\bw$, where $\bw$ corresponds to the subgraph of $G$ consisting of $\psi$-monochromatic edges only (cf.~\Cref{thm:ldp}), i.e., $\bw$ is obtained from $\bx$ by zeroing out all entries except those indexed by $\{u,v\}$ with $\psi(u) = \psi(v)$. We choose the parameter $t = 2s$, where $s \ge Cn\log n$ is the sparsity parameter from \Cref{thm:ldp}, which gives $d = O(s\log n)$; we choose $\gamma = 1/4$ and $\delta = 1/n$, giving $d' = O(\log^3 n)$.

Notice that \Cref{alg:streaming} requires a guess for $\degen := \degen(G)$, which is not known in advance. Our final one-pass algorithm runs $O(\log n)$ parallel instances of $\proc{Color}(\sigma,k)$, using geometrically spaced guesses $k = 2,4,8\ldots\,$. It outputs the coloring produced by the non-aborting run that uses the smallest guess.

\begin{theorem} \label{thm:color_kappa}
  Set $s = \ceil{\eps^{-2} n\log n}$, where $\eps > 0$ is a parameter. The above one-pass algorithm processes a dynamic (i.e., turnstile) graph stream using $O(\eps^{-2} n\log^4 n)$ bits of space and, with high probability, produces a proper coloring using at most $(1+O(\eps))\degen$ colors. In particular, taking $\eps = 1/\log n$, it produces a $\degen + o(\degen)$ coloring using $\tO(n)$ space. Each edge update is processed in $\tO(1)$ time and post-processing at the end of the stream takes $\tO(n)$ time.
\end{theorem}
\begin{proof}
  The coloring produced is obviously proper. Let us bound the number of colors used. One of the parallel runs of $\proc{Color}(\sigma,k)$ in \ref{alg:streaming} will use a value $k = k^\star \in (\degen,2\degen]$. We shall prove that, w.h.p., (a)~every non-aborting run with $k \le k^\star$ will use at most $(1+O(\eps))\degen$ colors, and (b)~the run with $k = k^\star$ will not abort.
  
  We start with~(a). Consider a particular run using $k \le k^\star$. By \cref{ldp:degen} of \Cref{thm:ldp}, each $G_i$ has degeneracy at most $(\degen+\lambda)/\ell$; so if $\bw$ is correctly recovered by the sparse recovery sketch (i.e., $\bw' = \bw$ in \Cref{alg:streaming}), then each $G_i$ is correctly recovered and the run uses at most $\degen+\lambda+\ell$ colors, as in \cref{eq:numcolors}. Using the values from \cref{eq:ldef}, this number is at most $(1+O(\eps))\degen$. Now, if the run does not abort, then the estimate of the sparsity $\|\bw\|_0$ is at most $5s/4$. By the guarantees of the $\ell_0$-estimation sketch, the true sparsity is at most $(5/4)(5s/4) < 2s = t$, so, w.h.p., $\bw$ is indeed $t$-sparse and, by the guarantees of the sparse recovery sketch, $\bw' = \bw$. Taking a union bound over all $O(\log n)$ runs, the bound on the number of colors holds for all required runs simultaneously, w.h.p.
  
  We now take up~(b). Note that $\|\bw\|_0$ is precisely the number of $\psi$-monochromatic edges in $G$. By \cref{ldp:sparse} of \Cref{thm:ldp}, we have $\|\bw_0\| \le s$ w.h.p. By the accuracy guarantee of the $\ell_0$-estimation sketch, in this run the estimate of $\|\bw\|_0$ is at most $5s/4$ w.h.p., so the run does not abort.
  
  The space usage of each parallel run is dominated by the computation of $\by$, so it is $O(d\log m) = O(s\log n \log m) = O(\eps^{-2} n \log^3 n)$, using our setting of $s$ and the assumption $m = \poly(n)$. The claims about the update time and post-processing time follow directly from the properties of a state-of-the-art sparse recovery scheme, e.g., the scheme based on expander matching pursuit given in Theorem~9 of Gilbert and Indyk~\cite{GilbertI10}.
\end{proof}

\mypar{Simplification for Insertion-Only Streams}
\Cref{alg:streaming} can be simplified considerably if the input stream is insertion-only. We can then initialize each $G_i$ to an empty graph and, upon seeing an edge $\{u,v\}$ in the stream, insert it into $G_i$ iff $\psi(u) = \psi(v) = i$. We abort if we collect more than $s$ edges; w.h.p., this will not happen, thanks to \Cref{thm:ldp}. Finally, we color the collected graphs $G_i$ greedily, just as in \Cref{alg:streaming}. With this simplification, the overall space usage drops to $O(s\log n) = O(\eps^{-2} n \log^2 n)$ bits. 

The reason this does not work for dynamic graph streams is that the number of monochromatic edges could exceed $s$ by an arbitrary amount mid-stream.


\subsection{Query Model} \label{sec:query-ub}

We now turn to the {\em general graph query model}, a standard model of space-conscious 
algorithms for big graphs where the input graph is random-accessible but the emphasis is on the 
examining only a tiny (ideally, sublinear) portion of it; for general background see Chapter~10 of
Goldreich's book~\cite{Goldreich-proptest-book}. In this model, the algorithm starts out knowing
the vertex set $[n]$ of the input graph $G$ and can access $G$ only through the following types of queries.
\begin{itemize}
  \item A {\em pair query}\/ $\Pair(\{u,v\})$, where $u,v \in [n]$. The query returns $1$ if $\{u,v\} \in E(G)$
  and $0$ otherwise. For better readability, we shall write this query as $\Pair(u,v)$.
  \item A {\em neighbor query}\/ $\Nbr(u,j)$, where $u \in [n]$ and $j \in [n-1]$. The query returns $v \in [n]$
  where $v$ is the $j$th neighbor of $u$ in some underlying fixed ordering of vertex adjacency lists; if
  $\deg(v) < j$, so that there does not exist a $j$th neighbor, the query returns $\bot$.
\end{itemize}
Naturally, when solving a problem in this model, the goal is to do so while minimizing the number of queries.

By adapting the combinatorial machinery from their semi streaming algorithm, Assadi \etal\cite{AssadiCK19}
gave an $\tO(n^{3/2})$-query algorithm for finding a $(\Delta+1)$-coloring. Our LDP framework gives a 
considerably simpler algorithm using $\degen + o(\degen)$ colors, where $\degen := \degen(G)$.
We remark here that $\tO(n^{3/2})$ query complexity is essentially optimal, as Assadi \etal\cite{AssadiCK19}
proved a matching lower bound for any $(c\cdot \Delta)$-coloring algorithm, for any constant $c>1$.
\begin{theorem} \label{thm:query}
  Given query access to a graph $G$, there is a randomized algorithm that, with high probability,
  produces a proper coloring of $G$ using $\degen + o(\degen)$ colors.
  The algorithm's worst-case query complexity, running time, and space usage are all $\tO(n^{3/2})$.
\end{theorem}
\begin{proof}
  The algorithm proceeds in two stages. In the first stage, it attempts to extract all edges in $G$ through neighbor queries alone, aborting when ``too many'' queries have been made. More precisely, it loops over all vertices $v$ and, for each $v$, issues queries $\Nbr(v,1), \Nbr(v,2), \ldots$ until a query returns $\bot$. If this stage ends up making $3n^{3/2}$ queries (say) without having processed every vertex, then it aborts and the algorithm moves on to the second stage. By \Cref{fact:degen-sparse}, if $\degen \le \sqrt n$, then this stage will not abort and the algorithm will have obtained $G$ completely; it can then $(\degen+1)$-color $G$ (as in \Cref{lem:degen-color}) and terminate, skipping the second stage.
  
  In the second stage, we know that $\degen > \sqrt n$. The algorithm now uses a random coloring $\psi$ to construct an $(\ell,s,\lambda)$-LDP of $G$ using the ``guess'' $k = \sqrt n$, with $s = \Theta(\eps^{-2}n\log n)$ and $\ell, \lambda$ given by \Cref{eq:ldef}. To produce each subgraph $G_i$ in the LDP, the algorithm simply makes all possible queries $\Pair(u,v)$ where $\psi(u) = \psi(v)$. W.h.p., the number of queries made is at most
  \[
    \frac12\sum_{i\in[\ell]} |V(G_i)|^2 \le \frac{\ell}{2} \left(\frac{2n}{\ell}\right)^2 \le \frac{2n^2 s}{4nk} = \Theta\left( \frac{n^{3/2} \log n}{\eps^2} \right) \,,
  \]
  where the first inequality uses \Cref{ldp:block} of \Cref{thm:ldp}. We can enforce this bound in the worst case by aborting if it is violated.
  
  Clearly, $k \le 2\degen$, so \Cref{ldp:degen} of \Cref{thm:ldp} applies and by the discussion after \Cref{def:ldp}, the algorithm uses $(1+O(\eps))\degen$ colors. Setting $\eps = 1/\log n$, this number is at most $\degen + o(\degen)$ and the overall number of queries remains $\tO(n^{3/2})$, as required.
\end{proof}


\subsection{MPC and Congested Clique Models}

In the Massively Parallel Communication
(MPC) model of Beame \etal\cite{BeameKS13},
an input of size $m$ is distributed adversarially among $p$
processors, each of which has $S$ bits of working memory: here, $p$ and $S$ are $o(m)$
and, ideally, $p \approx m/S$.
Computation proceeds
in synchronous rounds: in each round, 
a processor carries out some local
computation (of arbitrary time complexity) and then communicates with as many of the other processors as desired, provided that each processor sends and receives no more than $S$ bits per round. The primary goal in solving a problem is to minimize
the number of rounds.

When the input is an $n$-vertex graph, the most natural and widely studied setting of MPC is $S = \tO(n)$, which enables each processor to hold some information about
every vertex; this makes many graph problems tractable.
Since the input size $m$ is potentially $\Omega(n^2)$, it is reasonable to allow $p = n$ many processors. Note that the input is just a collection of edges, distributed adversarially among these processors, subject to the memory constraint.

\begin{theorem} \label{thm:MPCmain}
There is a randomized $O(1)$-round MPC algorithm that, given an $n$-vertex graph $G$, outputs a $(\degen+o(\degen))$-coloring of $G$ with high probability. The algorithm uses $n$ processors, each with $O(n \log^2 n)$ bits of memory.
\end{theorem}
\begin{proof}
  Our algorithm will use $n$ processors, each assigned to one vertex.
  If $|E(G)| = O(n\log n)$, then all of $G$ can be collected at one processor in a single round using $|E(G)|\cdot 2\ceil{\log n} = O(n\log^2 n)$ bits of communication and the problem is solved trivially. Therefore, we may as well assume that $|E(G)| = \omega(n\log n)$, which implies $\degen = \omega(\log n)$, by \Cref{fact:degen-sparse}.
  We shall first give an algorithm assuming that $\degen$ is known {\em a priori}. Our final algorithm will be a refinement of this preliminary one. 
  
  \medskip \noindent
  \emph{Preliminary algorithm.~} Take $k = \degen$. Each processor chooses a random color for its vertex, implicitly producing a partition $(G_1, \ldots, G_\ell)$ that is, w.h.p., an ($\ell,s,\lambda)$-LDP; we take $\ell,\lambda$ as in \cref{eq:ldef}, $s = \Theta(\eps^{-2} n\log n)$, and $\eps = (k^{-1} \log n)^{1/4}$. Note that $\eps = o(1)$. In Round~1, each processor sends its chosen color to all others---this is $O(n\log n)$ bits of communication per machine---and as a result every processor learns which of its vertex's incident edges are monochromatic. Now each color $i \in [\ell]$ is assigned a unique machine $M_i$ and, in Round~2, all edges in $G_i$ are sent to $M_i$. Each $M_i$ then locally computes a $(\degen(G_i)+1)$-coloring of $G_i$ using a palette disjoint from those of other $M_i$s; by the discussion following \Cref{def:ldp}, this colors $G$ using at most $(1+O(\eps))\degen = \degen+o(\degen)$ colors.
  
  The communication in Round~2 is bounded by $\max_i |E(G_i)|\cdot 2\ceil{\log n}$. By \Cref{fact:degen-sparse}, \cref{ldp:degen,ldp:block} of \Cref{thm:ldp}, and \cref{eq:ldef}, the following holds w.h.p.~for each $i\in[\ell]$:
\begin{equation}
  \label{eq:block-sparse}
    |E(G_i)| \le \degen(G_i) |V(G_i)|
    \le \frac{\degen+\lambda}{\ell} \frac{2n}{\ell}
    \le \frac{4n\degen}{\ell^2}
    \le \frac{4nk}{(2nk/s)^2} 
    \ifthenelse{\equal{\acmconf}{1}} { \\ } {}
    = \frac{O(\eps^{-2}n\log n)^2}{nk}
    = O\left( \frac{n\log^2 n}{\eps^4 k} \right)
    = O(n \log n)\,.
    \end{equation}
  Thus, the communication per processor in Round~2 is $O(n\log^2 n)$ bits.
  
  \medskip \noindent 
  \emph{Final algorithm.~} When we don't know $\degen$ in advance, we can make geometrically spaced guesses $k$, as in \Cref{sec:streaming}. In Round~1, we choose a random coloring for each such $k$. In Round~2, we  determine the quantities $|E(G_i)|$ for each $k$ and each subgraph $G_i$ and thereby determine the smallest $k$ such that  \cref{eq:block-sparse} holds for every $G_i$ corresponding to this $k$. We then run Round~3 for only this one $k$, replicating the logic of Round~2 of the preliminary algorithm.
  
  Correctness is immediate. We turn to bounding the communication cost. For Round~3, the previous analysis shows that the communication per processor is $O(n\log ^2 n)$ bits. For Rounds~1 and~2, let us consider the communication involved for each guess $k$: since each randomly-chosen color and each cardinality $|E(G_i)|$ can be described in $O(\log n)$ bits, each processor sends and receives at most $O(n\log n)$ bits per guess. This is a total of $O(n\log^2 n)$ bits, as claimed.
\end{proof}

The CONGESTED-CLIQUE model~\cite{LotkerPPP05} is a well established model of distributed computing for graph problems.
In this model, there are $n$ nodes, each of which holds the local neighborhood information (i.e., the incident edges) of one vertex of the input graph $G$. In each round, every pair of
nodes may communicate, whether or not they are adjacent in $G$, but the communication is restricted to $O(\log n)$ bits. There is no constraint on a node's local memory. The goal is to minimize
the number of rounds.

Behnezhad \etal\cite{Behnezhad2018BriefAS} built on results of
Lenzen \cite{Lenzen:2013} to show that any
algorithm in the {\em semi-MPC model}---defined as MPC with
space per machine being $O(n\log n)$ bits---can be simulated in the Congested Clique model, preserving the round complexity up to
a constant factor. Based on this, we obtain the following result.

\begin{theorem} \label{thm:cong-clique}
  There is a randomized $O(1)$-round algorithm in the Congested Clique model that, given a graph $G$, w.h.p.~finds a
  $(\degen+O(\degen^{3/4}\log^{1/2}n))$-coloring. 
  For $\degen = \omega(\log^2 n)$,
  this gives a $(\degen+o(\degen))$-coloring.\qed
\end{theorem}
\begin{proof}
  We cannot directly use our algorithm in \Cref{thm:MPCmain} because it is not a semi-MPC algorithm: it uses $O(n\log^2 n)$ bits of space per processor, rather than $O(n\log n)$. However, with a more efficient implementation of Round~1, a more careful analysis of Round~2, and a slight tweak of parameters for Round~3, we can improve the communication (hence, space) bounds to $O(n\log n)$, whereupon the theorem of Behnezhad \etal\cite{Behnezhad2018BriefAS}~completes the proof.
  
  For Round~3, the tweak is to set $\eps = (k^{-1} \log^2 n)^{1/4}$ but otherwise replicate the logic of the final algorithm from \Cref{thm:MPCmain}. With this higher value of $\eps$, the bound from \cref{eq:block-sparse} improves to $|E(G_i)| = O(n)$. Therefore the per-processor communication in Round~3 is only $O(n\log n)$ bits. The number of colors used is, w.h.p., at most $(1+O(\eps))\degen = \degen+O(\degen^{3/4}\log^{1/2}n)$. 

  For a tighter analysis of the communication cost of Round~2, note that, for a particular guess $k$, there is a corresponding $\ell$ given by \cref{eq:ldef} such that each processor need only send/receive $\ell$ cardinalities $|E(G_i)|$, each of which can be described in $O(\log n)$ bits. Consulting \cref{eq:ldef}, we see that $\ell = O(n^2/s) = O(n/\log n)$. Therefore, summing over all $O(\log n)$ choices of $k$, each processor communicates at most
  \[
    O(n/\log n)\cdot O(\log n)\cdot O(\log n) = O(n\log n) \text{ bits.}
  \]

   Round~1 appears problematic at first, since there are $O(\log n)$ many random colorings to be chosen, one for each guess $k$. However, note that these colorings need not be independent. Therefore, we can choose just one random $\ceil{\log n}$-bit ``master color'' $\phi(v)$ for each vertex $v$ and derive the random colorings for the various guesses $k$ by using only appropriate length prefixes of $\phi(v)$. This ensures that each processor only communicates $O(n\log n)$ bits in Round~1. 
\end{proof}
  


\subsection{Distributed Coloring in the LOCAL Model} \label{sec:local}

In the LOCAL model, each node of the input graph $G$ hosts
a processor that knows only its own neighborhood. The
processors operate in synchronous rounds, during which
they can send and receive messages of arbitrary length
to and from their neighbors. The processors are 
allowed unbounded local computation in each round.
The key complexity measure is {\em time}, defined as the
number of rounds used by an algorithm (expected number,
for a randomized algorithm) on a worst-case input.

Graph coloring in the LOCAL model
is very heavily studied and is one of {\em the} central problems in
distributed algorithms. Here, our focus is on algorithms that properly color
the input graph $G$ using a number of colors that depends on $\alpha := \alpha(G)$,
the arboricity of $G$. Recall that $\alpha \le \degen \le 2\alpha-1$ (\Cref{fact:degen-arb}).
Unlike in previous sections, our results will give big-$O$ bounds on the number of colors, so we may as well
state them in terms of $\alpha$ (following established tradition in this line of work)
rather than $\degen$. Our focus will be on algorithms that run in {\em sublogarithmic}
time, while using not too many colors. See \Cref{sec:related} for a quick summary
of other interesting parameter regimes and Barenboim and Elkin~\cite{BarenboimElkin-book}
for a thorough treatment of graph coloring in the LOCAL model.

Kothapalli and Pemmaraju~\cite{Kothapalli2011distributed} gave
an $O(k)$-round algorithm that uses $O(\alpha n^{1/k})$ colors,
for all $k$ with $2\log \log n \leq k \leq \sqrt{\log n}$.
We give a new coloring algorithm that, in particular, extends the range of $k$ to which such a time/quality
tradeoff applies: for $k \in \big[ \omega(\sqrt{\log n}),\, o(\log n) \big]$,
we can compute an $O(\alpha n^{1/k} \log n)$-coloring in $O(k)$ rounds.

Our algorithm uses our LDP framework to split the input graph into parts with logarithmic degeneracy
(hence, arboricity) and then invokes an algorithm of Barenboim and Elkin. The following
theorem records the key properties of their algorithm.

\begin{lemma}[Thm 5.6 of Barenboim and Elkin~\cite{Barenboim2010sublogarithmic}] \label{thm:dist_arb}
    There is a deterministic distributed algorithm in the LOCAL model that, given an $n$-vertex graph $G$,
    an upper bound $b$ on $\alpha(G)$, and a parameter $t$ with $2 < t \le O( \sqrt{{n}/b})$,
    produces an $O(tb^2)$-coloring of $G$ in time $O\left({\log_{t} n} + \log^{\star}n \right)$.
    \qed
\end{lemma}

Here is the main result of this section.
\begin{theorem}
\label{thm:color_dist}
    There is a randomized distributed algorithm in the LOCAL model that, given an $n$-vertex graph $G$,
    an estimate of its arboricity $\alpha$ up to a constant factor, and a parameter $t$ such that $2 < t \le O(\sqrt{n /\log n})$,
    produces an $O(t\alpha\log n)$-coloring of $G$ in time $O\left( {\log_{t} n} +\log^{\star} n \right)$.
\end{theorem}
\begin{proof}
    To simplify the presentation, we assume that $\alpha = \alpha(G)$. We assume that every node (vertex) knows $n$ and $\alpha$.
    Consider a $(\ell,s,\lambda)$-LDP of $G$, where we put $s=Cn\log n$, for some large
    constant $C$, as in \Cref{thm:ldp}.
    This setting of $s$ gives $\ell = O(\alpha / \log n)$.
    First, each vertex $v$ chooses a color $\psi(v)$ uniformly at random from $[\ell]$.
    Next, we need to effectively ``construct'' the blocks $G_i$, for each $i \in [\ell]$.
    This is straightforwardly done in a single round: each vertex $v$ sends $\psi(v)$ to all its neighbors.
    
    At this point, each vertex $v$ knows its neighbors in the block $G_{\psi(v)}$. So it's 
    now possible to run a distributed algorithm on each $G_i$. We invoke the algorithm in
    \Cref{thm:dist_arb}. The algorithm needs each vertex $v$ to know an upper bound $b_i$ on
    $\alpha(G_i)$, where $i = \psi(v)$. A useful upper bound of $b_i = O(\log n)$, which holds w.h.p.,
    is given by \cref{ldp:degen} of \Cref{thm:ldp}.

    By \Cref{thm:dist_arb}, each $G_i$ can be colored using $O(t\log^2 n)$ colors, within
    another $O\left({\log_{t} n} + \log^{\star}n \right)$ rounds, since $2 < t \leq O(\sqrt{n/\log n})$.
    Using disjoint palettes for the distinct blocks, the total number of colors used for $G$
    is at most $\ell \cdot O(t\log^2 n) = O(t\alpha \log n)$, as required.
\end{proof}

The particular form of the tradeoff stated in \Cref{table:results} is obtained by setting $t=n^{1/k}$ (for some $k \ge 3$) in the above theorem. 
\begin{corollary} \label{cor:colors_dist}
  There is a randomized LOCAL algorithm that, given graph $G$, estimate $\alpha \approx \alpha(G)$,
  and a parameter $k$ with $2 < n^{1/k} \le O(\sqrt{n/ \log n})$, finds an
  $O(\alpha n^{1/k}\log n)$-coloring of $G$ in time $O\left( k +\log^{\star} n \right)$.
\qed
\end{corollary}



\section{Lower Bounds} \label{sec:lb}

Can we improve the guarantees of our algorithms so that they use at most $\degen+1$ colors, rather than $\degen+o(\degen)$? After all, every graph $G$ does have a proper $(\degen(G)+1)$-coloring. The main message of this section is that answer is a strong ``No,'' at least in the data streaming and query models. If we insist on a coloring that good, we would incur the worst possible space or query complexity: $\Omega(n^2)$. In fact, this holds even if $\degen$ is known to the algorithm in advance. Moreover, all our streaming lower bounds hold even if the input stream consists of edge insertions alone.

Our lower bounds generalize to the problem of producing a $(\degen+\lambda)$-coloring. We show that this requires $\Omega(n^2/\lambda^2)$ space or query complexity. Such generalizations are based on the following Blow-Up Lemma.

\begin{definition} \label{def:blowup}
  Let $G$ be a graph and $\lambda$ a positive integer. The {\em blow-up graph} $G^\lambda$ is obtained by replacing each vertex of $G$ with a copy of the complete graph $K_\lambda$ and replacing each edge of $G$ with a complete bipartite graph between the copies of $K_\lambda$ at its endpoints. More succinctly, $G^\lambda$ is the lexicographical product $G[K_\lambda]$.
\end{definition}
  
\begin{lemma}[Blow-Up Lemma] \label{lem:blowup}
  For all graphs $G$ and positive integers $\lambda,c$, if $G$ has a $c$-clique, then $G^\lambda$ has a $(c\lambda)$-clique.
  Also, $\degen(G^\lambda) \le (\degen(G)+1)\lambda -1$.
\end{lemma}
\begin{proof}
  The claim about cliques is immediate. The bound on $\degen(G^\lambda)$ follows by taking a degeneracy ordering of $G$ and replacing each vertex $v$ by a list of vertices of the clique that replaces $v$ in $G^\lambda$, ordering vertices within the clique arbitrarily.
\end{proof}

Our lower bounds come in two flavors. The first address the hardness of distinguishing low-degeneracy graphs from high-chromatic-number graphs. This is encapsulated in the following abstract problem.
\begin{definition}[\gdct problem] \label{def:graph-dist}
  Consider two graph families: $\cG_1 := \cG_1(n,q,\lambda)$, consisting of $n$-vertex graphs with chromatic number $\chi \geq (q+1)\lambda$,
  and $\cG_2 := \cG_2(n,q,\lambda)$, consisting of  $n$-vertex graphs with $\degen \leq q\lambda - 1$.
  Then $\gdct(n,q,\lambda)$ is the problem of distinguishing $\cG_1$ from $\cG_2$;
  note that $\cG_1 \cap \cG_2 = \varnothing$. More precisely, given an input graph $G$ on $n$ vertices,
  the problem is to decide whether $G \in \cG_1$ or $G \in \cG_2$, with success probability at least $2/3$.
\end{definition}
We shall prove that \gdct is ``hard'' in the insertion-only streaming setting and in the query setting, thereby establishing that in these models it is hard to produce a $(\degen+\lambda)$-coloring. In fact, our proofs will show that it is just as hard to estimate the parameter $\degen$; this goes to show that the hardness of the coloring problem is not just because of the large output size.

Lower bounds of the above flavor raise the following question: since estimating $\degen$ itself is hard, does the coloring problem become easier if the value of $\degen(G)$ is given in advance, before the algorithm starts to read $G$? In fact, the $(\Delta+1)$-coloring algorithms by Assadi \etal~\cite{AssadiCK19} assume that $\Delta$ is known in advance. However, perhaps surprisingly, we prove a second flavor of lower bounds, showing that {\em a priori} knowledge of $\degen$ does not help and $(\degen+1)$-coloring (more generally, $(\degen+\lambda)$-coloring) remains a hard problem even under the strong assumption that $\degen$ is known in advance. 


\begin{figure*} \label{fig:diagram}
\centering
\begin{subfigure}{0.3\textwidth}
    \centering
  \includegraphics[height=2in, width=1.7in]{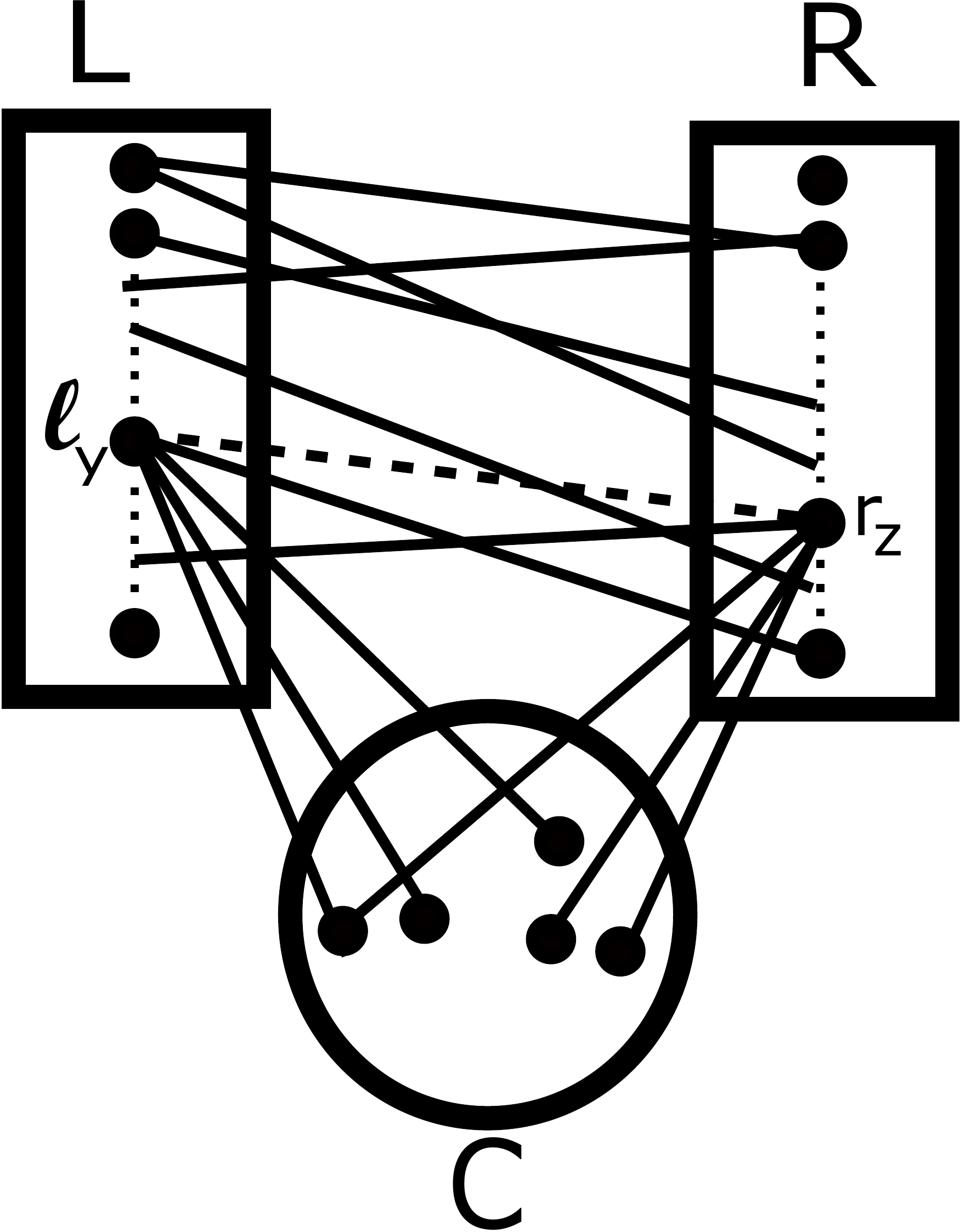}
  \caption{}\label{fig:1}
\end{subfigure}\hfill
\begin{subfigure}{0.3\textwidth}
    \centering
  \includegraphics[height=2.4in, width=1.5in]{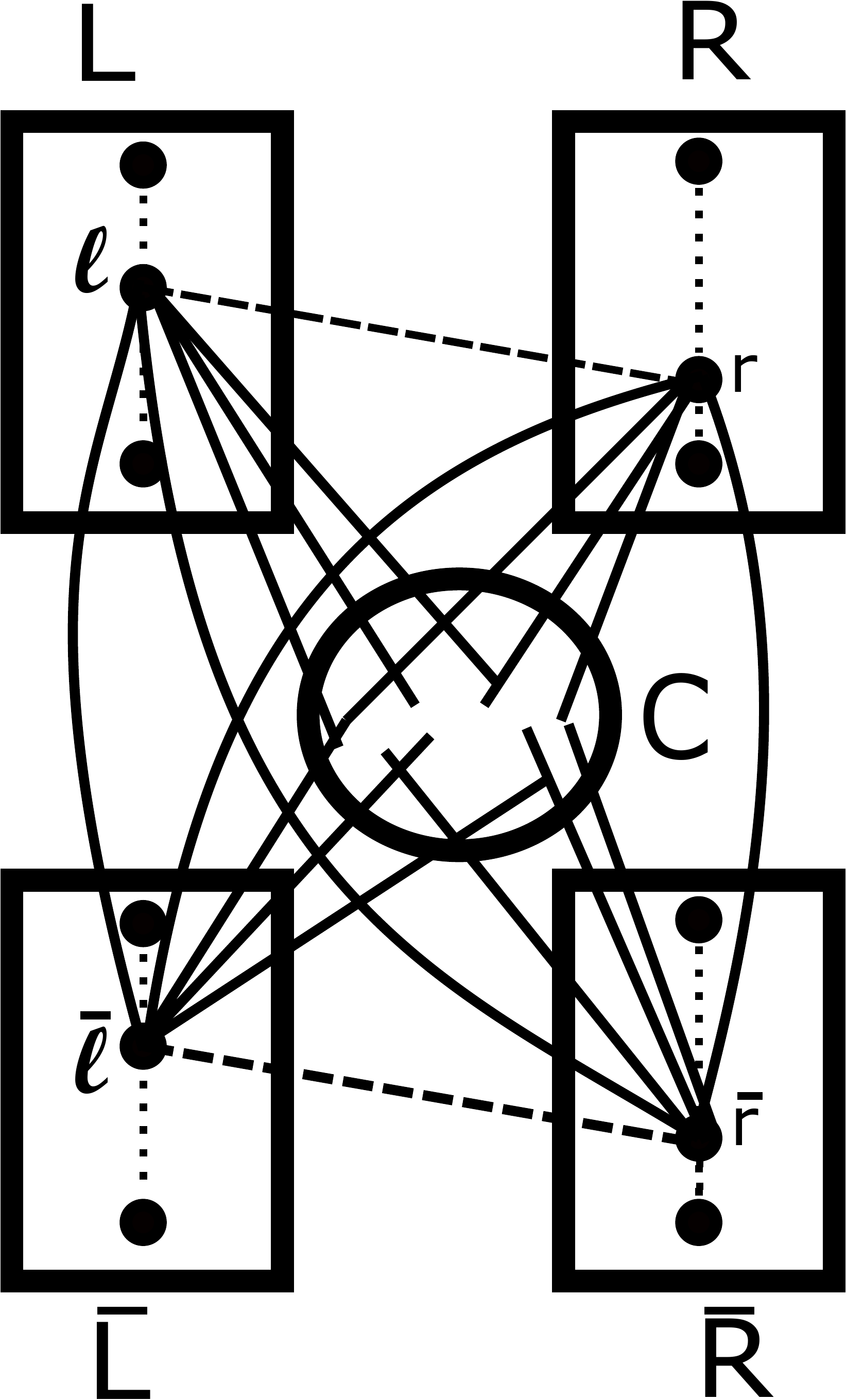}
  \caption{}\label{fig:2}
\end{subfigure}\hfill
\begin{subfigure}{0.3\textwidth}%
    \centering
  \includegraphics[height=1.5in, width=1.7in]{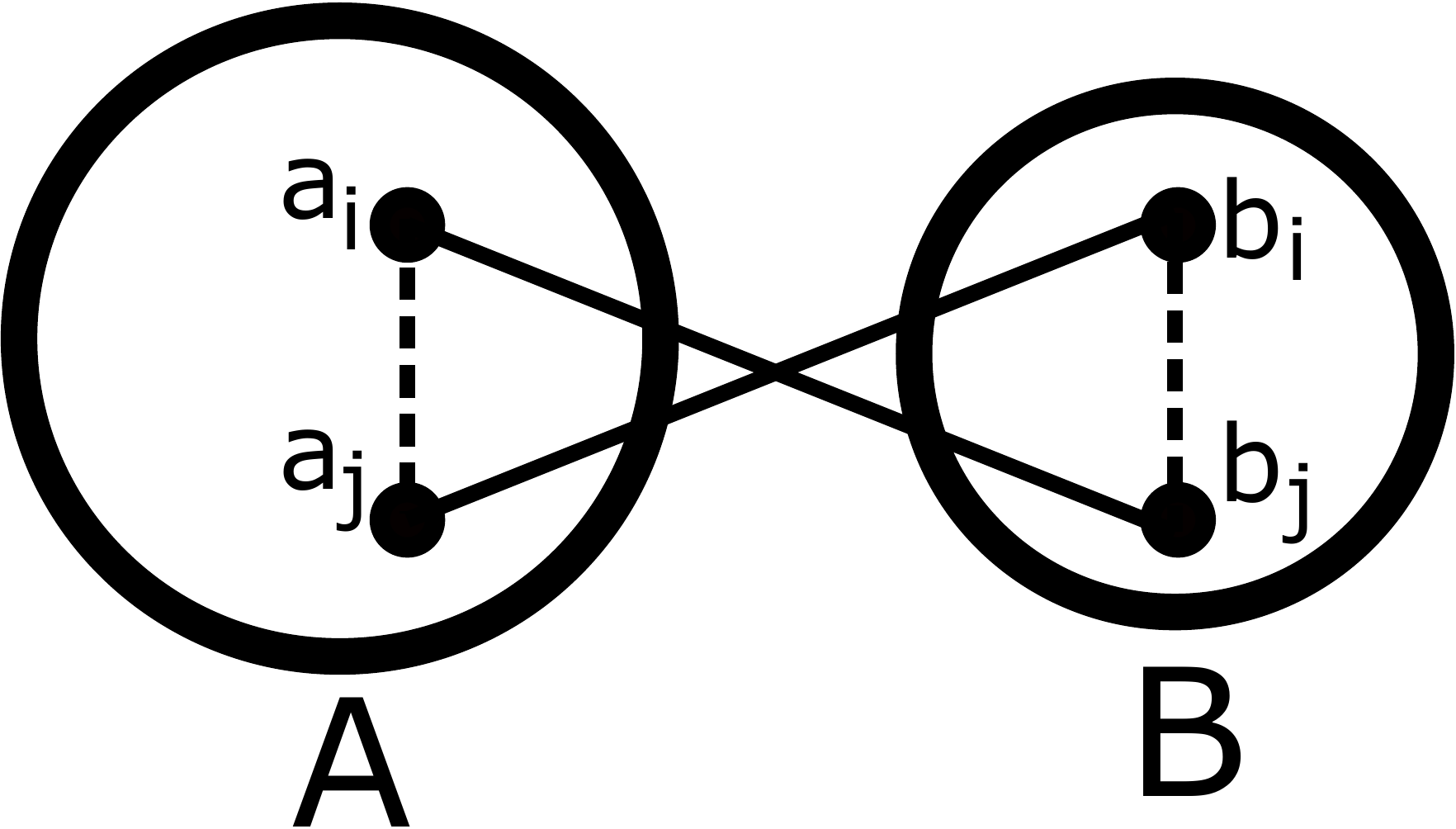}
  \caption{}\label{fig:3}
\end{subfigure}
\caption{Gadget graphs used in (a) \Cref{lem:streaming-lb-distinguish}; 
(b) \Cref{thm:streamcoloring-lb};
(c) \Cref{lem:query-lb-distinguish} and \Cref{thm:query-colorlb-lambda}.}
\end{figure*}

\subsection{Streaming Lower Bounds}

In this section, we prove both flavors of lower bounds in the one-pass streaming setting. The next section takes up the query model.

Our streaming lower bounds use reductions from the \indx and \intfind (intersection finding, a variant of \textsc{disjointness}) problems in communication complexity. In the $\indx_N$ problem, Alice is given a vector $\bx = (x_1, \ldots, x_N) \in \b^N$ and Bob is given an index $k\in[N]$. The goal is for Alice to send Bob a (possibly random) $c$-bit message that enables Bob to output $x_k$ with probability at least $2/3$. The smallest $c$ for which such a protocol exists is called the one-way randomized communication complexity, $\R^\to(\indx_N)$. In $\intfind_N$, Alice and Bob hold vectors $\bx,\by \in \{0,1\}^N$, interpreted as subsets of $[N]$, satisfying the promise that $|\bx\cap\by| = 1$. They must find the unique index $i$ where $x_i = y_i = 1$, using at most $c$ bits of randomized interactive communication, succeeding with probability at least $2/3$. The smallest $c$ for which such a protocol exists is the randomized communication complexity, $\R(\intfind_N)$.
As is well known, $\R^\to(\indx_N) = \Omega(N)$~\cite{Ablayev96} and $\R(\intfind_N) = \Omega(N)$; the latter is a simple extension of the \textsc{disjointness} lower bound~\cite{Razborov92}.

We shall in fact consider instances of $\indx_N$ where $N = p^2$, for an integer $p$. Using a canonical bijection between $[N]$ and $[p] \times [p]$, we reinterpret $\bx$ as a matrix with entries $(x_{ij})_{i,j \in [p]}$, and Bob's input as $(y,z) \in [p] \times [p]$. We further interpret this matrix $\bx$ as the bipartite adjacency matrix of a $(2p)$-vertex balanced bipartite graph $H_\bx$. Such graphs $H_\bx$ will be key gadgets in the reductions to follow.

\begin{definition} \label{def:gadget-graph}
  For $\bx \in \b^{p \times p}$, a realization of $H_\bx$ on a list $(\ell_1, \ldots, \ell_p, r_1, \ldots, r_p)$ of distinct vertices is a graph on these vertices whose edge set is $\{\{\ell_i, r_j\}:\, x_{ij} = 1\}$. 
\end{definition}

\mypar{First Flavor: Degeneracy Not Known in Advance} 
To prove lower bounds of the first flavor, we start by demonstrating the hardness
of the abstract problem $\gdct$, from \Cref{def:graph-dist}.
\begin{lemma} \label{lem:streaming-lb-distinguish}
  Solving $\gdct(n,q,\lambda)$ in one randomized streaming pass requires $\Omega(n^2/\lambda^2)$ space.
  
  More precisely, there is a constant $c > 0$ such that for every integer $\lambda \ge 1$ and every sufficiently
  large integer $q$, there is a setting $n = n(q,\lambda)$ for which every randomized one-pass streaming algorithm 
  for $\gdct(n,q,\lambda)$ requires at least $c n^2/\lambda^2$ bits of space.
\end{lemma}
\begin{proof}
  Put $p = q-1$. We reduce from $\indx_N$, where $N = p^2$, using the following plan.
  Starting with an empty graph on $n = 3 \lambda p$ vertices, Alice adds certain edges based on her input
  $\bx \in \b^{p \times p}$ and then Bob adds certain other edges based on his
  input $(y,z) \in [p] \times [p]$. By design, solving $\gdct(n,q,\lambda)$ on the resulting final graph
  reveals the bit $x_{yz}$, implying that a one-pass streaming algorithm for $\gdct$
  requires at least $\R^\to(\indx_N) = \Omega(N) = \Omega(p^2) = \Omega(n^2/\lambda^2)$ bits of memory. The 
  details follow.
  
  We first consider $\lambda = 1$. We use the vertex set $L \uplus R \uplus C$
  (the notation ``$\uplus$'' denotes a disjoint union), where $L = \{\ell_1, \ldots, \ell_p\}$,
  $R = \{r_1, \ldots, r_p\}$, and $|C| = p$. Alice introduces
  the edges of the gadget graph $H_\bx$ (from \Cref{def:gadget-graph}), realized on the vertices
  $(\ell_1, \ldots, \ell_p, r_1, \ldots, r_p)$. Bob introduces all possible edges within $C \cup \{\ell_y, r_z\}$, except for $\{\ell_y, r_z\}$.
  Let $G$ be the resulting graph.
  
  If $x_{yz} = 1$, then $G$ contains a clique on $C \cup \{\ell_y, r_z\}$, whence $\chi(G) \ge p+2$.
  If, on the other hand, $x_{yz} = 0$, then we claim that $\degen(G) \le p$.
  By \Cref{lem:degen-odeg}, the claim will follow if we exhibit a vertex ordering $\before$
  such that $\odeg_{G,\before}(v) \le p$ for all $v \in V(G)$. We use an ordering where
  \[
    L \cup R \setminus \{\ell_y, r_z\} \before \ell_y \before \{r_z\} \cup C
  \]
  and the ordering within each set is arbitrary. By construction of $H_\bx$, each vertex
  in $L \cup R \setminus \{\ell_y, r_z\}$ has {\em total} degree at most $p$. For each
  vertex $v \in \{r_z\} \cup C$, we trivially have $\odeg_{G,\before}(v) \le p$ because
  $|C| = p$. Finally, since $x_{yz} = 0$, the vertex $r_z$ is not a neighbor of $\ell_y$;
  so $\odeg_{G,\before}(\ell_y) = |C| = p$. This proves the claim.
  
  When $\lambda \ge 1$, Alice and Bob introduce edges so as to create the blow-up 
  graph $G^\lambda$, as in \Cref{def:blowup}. By \Cref{lem:blowup}, if $x_{yz} = 1$,
  then $G^\lambda$ has a $(p+2)\lambda$-clique, whereas if $x_{yz} = 0$, then
  $\degen(G^\lambda) \le (p+1)\lambda - 1$. In the former case, $\chi(G^\lambda) \ge
  (p+2)\lambda = (q+1)\lambda$, so that $G^\lambda \in \cG_1(n,q,\lambda)$; cf.~\Cref{def:graph-dist}. In the
  latter case, $\degen(G^\lambda) \le q\lambda - 1$, so that $G^\lambda \in \cG_2(n,q,\lambda)$.
  Thus, solving $\gdct(n,q,\lambda)$ on $G^\lambda$ reveals $x_{yz}$.
\end{proof}

Our coloring lower bounds are straightforward consequences of the above lemma.

\begin{theorem} \label{thm:stream-lb}
  Given a single randomized pass over a stream of edges of an $n$-vertex graph $G$, succeeding with probability at least $2/3$ at either of the following tasks requires $\Omega(n^2/\lambda^2)$ space, where $\lambda \ge 1$ is an integer parameter:
\begin{enumerate}[label=(\roman*), itemsep=1pt]
    \item \label{task:color-stream} produce a proper $(\degen+\lambda)$-coloring of $G$;
    \item \label{task:estim-stream} produce an estimate $\hat{\degen}$ such that $|\hat{\degen} - \degen| \le  \lambda$.
\end{enumerate}
Furthermore, if we require $\lambda = O\bigl(\degen^{\frac12-\gamma}\bigr)$, where $\gamma > 0$, then neither task admits a semi-streaming algorithm.
\end{theorem}
\begin{proof}
  An algorithm for either task~\ref{task:color-stream} and or task~\ref{task:estim-stream} immediately solves $\gdct$ with appropriate parameters, implying the $\Omega(n^2/\lambda^2)$ bounds, thanks to \Cref{lem:streaming-lb-distinguish}. For the ``furthermore'' statement, note that the graphs in the family $\cG_2$ constructed in the proof of \Cref{lem:streaming-lb-distinguish} have $\degen = \Theta(n)$, so performing either task with the stated guarantee on $\lambda$ would require $\Omega(n^{1+2\gamma})$ space, which is not in $\tO(n)$.
\end{proof}

Combining the above result with the algorithmic result in \Cref{thm:color_kappa}, we see that producing a $(\degen + o(\degen))$-coloring is possible in semi-streaming space whereas producing a $(\degen + O\bigl(\degen^{\frac12-\gamma}\bigr))$-coloring is not. We leave open the question of whether this gap can be tightened.

%

\mypar{Second Flavor: Degeneracy Known in Advance}
We now show that the coloring problem remains just as hard even if 
the algorithm knows the degeneracy of the graph before seeing the edge stream.

\begin{theorem} \label{thm:streamcoloring-lb}
  Given as input an integer $\degen$, followed by a stream of edges of an $n$-vertex
  graph $G$ with degeneracy $\degen$, a randomized one-pass algorithm that 
  produces a proper $(\degen+\lambda)$-coloring of $G$ requires $\Omega(n^2/\lambda^2)$ bits of space.
  Furthermore, if we require $\lambda = O\bigl(\degen^{\frac12-\gamma}\bigr)$, where $\gamma > 0$, then the task does not admit a semi-streaming algorithm.

\end{theorem}
\begin{proof}
  We reduce from $\indx_N$, where $N = p^2$, using a plan analogous to the one 
  used in proving \Cref{lem:streaming-lb-distinguish}. Alice and Bob will construct a
  graph on $n = 5\lambda p$ vertices, using their respective inputs $\bx \in \b^{p\times p}$
  and $(y,z) \in [p]\times [p]$.

  First, we consider the case $\lambda = 1$. We use the vertex set $L \uplus R \uplus \Lbar \uplus \Rbar \uplus C$,
  where $L = \{\ell_1, \ldots \ell_p\}$, $R = \{r_1, \ldots, r_p\}$, $\Lbar = \{\lbar_1, \ldots, \lbar_p\}$, $\Rbar = \{\rbar_1, \ldots, \rbar_p\}$,
  and $|C| = p$. Let $\overline{\bx}$ be the bitwise complement of $\bx$.
  Alice introduces the edges of the gadget graph $H_\bx$ (from \Cref{def:gadget-graph}), realized on
  $L \cup R$, and the edges of $H_{\overline{\bx}}$ realized on $\Lbar \cup \Rbar$.
  For ease of notation, put $\ell := \ell_y$, $r := r_z$, $\lbar := \lbar_y$, $\rbar := \rbar_z$,
  and $S := C \cup \{\ell, r, \lbar, \rbar\}$.
  Bob introduces all possible edges within $S$, except for $\{\ell,r\}$
  and $\{\lbar,\rbar\}$. Let $G$ be the resulting graph.
  
  We claim that the degeneracy $\degen(G) = p+2$. To prove this, we consider the case $x_{yz} = 1$ (the other case, $x_{yz} = 0$, is symmetric). By construction, $G$ contains a clique on the $p+3$ vertices in $C \cup \{\ell, r, \lbar\}$; therefore, by definition of degeneracy, $\degen(G) \ge p+2$. To show that $\degen(G) \le p+2$, it will suffice to exhibit a vertex ordering $\before$ such that $\odeg_{G,\before}(v) \le p+2$ for all $v \in V(G)$. To this end, consider an ordering where
  \[
    V(G) \setminus S \before \lbar \before S \setminus \{\lbar\}
   \]
   and the ordering within each set is arbitrary. Each vertex $v \in V(G) \setminus S$ has $\odeg_{G,\before}(v) \le \deg(v) \le p$ and each vertex $v \in S \setminus \{\lbar\}$ has $\odeg_{G,\before}(v) \le \big|S \setminus \{\lbar\}\big| - 1 = p+2$. As for the vertex $\lbar$, since $\overline{x}_{yz} = 1 - x_{yz} = 0$, by the construction in \Cref{def:gadget-graph}, $\rbar$ is not a neighbor of $\lbar$; therefore, $\odeg_{G,\before}(\lbar) \le \big|S \setminus \{\lbar,\rbar\}\big| = p+2$.
   
   Let $\cA$ be a streaming algorithm that behaves as in the theorem statement. Recall that we are considering $\lambda=1$. Since $\degen(G) = p+2$ for every instance of $\indx_N$, Alice and Bob can simulate $\cA$ on their constructed graph $G$ by first feeding it the number $p+2$, then Alice's edges, and then Bob's. When $\cA$ succeeds, the coloring it outputs is a proper $(p+3)$-coloring; therefore it must repeat a color inside $S$, as $|S| = p+4$. But $S$ has exactly one pair of non-adjacent vertices: the pair $\{\ell,r\}$ if $x_{yz} = 0$, and the pair $\{\lbar,\rbar\}$ if $x_{yz}=1$. Thus, an examination of which two vertices in $S$ receive the same color reveals $x_{yz}$, solving the $\indx_N$ instance. It follows that $\cA$ must use at least $\R^\to(\indx_N) = \Omega(N) = \Omega(p^2)$ bits of space.

  Now consider an arbitrary $\lambda$. Alice and Bob proceed as above, except that they simulate $\cA$ on the blow-up graph $G^\lambda$. Since $G$ always has a $(p+3)$-clique and $\degen(G) = p+2$, the two halves of \Cref{lem:blowup} together imply $\degen(G^\lambda) = (p+3)\lambda - 1$. So, when $\cA$ succeeds, it properly colors $G^\lambda$ using at most $(p+4)\lambda -1$ colors. For each $A \subseteq V(G)$, abusing notation, let $A^\lambda$ denote its corresponding set of vertices in $G^\lambda$ (cf.~\Cref{def:blowup}). Since $|S^{\lambda}| = (p+4)\lambda$, there must be a color repetition within $S^\lambda$. Reasoning as above, this repetition must occur within $\{\ell,r\}^\lambda$ when $x_{yz} = 0$ and within $\{\lbar,\rbar\}^\lambda$ when $x_{yz} = 1$. Therefore, Bob can examine the coloring to solve $\indx_N$, showing that $\cA$ must use $\Omega(N) = \Omega(p^2) = \Omega(n^2/\lambda^2)$ space.

  The ``furthermore'' part follows by observing that $\degen(G^{\lambda}) = \Theta\bigl( |V(G^{\lambda})| \bigr)$.
\end{proof}

\mypar{Multiple Passes}
The streaming algorithm from \Cref{sec:streaming} is one-pass, as are the lower bounds proved above. Is the coloring problem any easier if we are allowed multiple passes over the edge stream? We now give a simple argument showing that, if we slightly generalize the problem, it stays just as hard using multiple ($O(1)$ many) passes.

The generalization is to allow some edges to be {\em repeated} in the stream. In other words, the input is a multigraph $\hat{G}$. Clearly, a coloring is proper for $\hat{G}$ iff it is proper for the underlying simple graph $G$, so the relevant algorithmic problem is to properly $(\degen+\lambda)$-color $G$, where $\degen := \degen(G)$. Note that our algorithm in \Cref{sec:streaming} does, in fact, solve this more general problem.

\begin{theorem} \label{thm:stream_multipass_lb}
  Given as input an integer $\degen$, followed by a stream of edges of an $n$-vertex
  multigraph $\hat{G}$ whose underlying simple graph has degeneracy $\degen$, a randomized $p$-pass algorithm that 
  produces a proper $(\degen+\lambda)$-coloring of $G$ requires $\Omega(n^2/(\lambda^2 p))$ bits of space. This holds even
  if the stream is insertion-only, with each edge appearing at most twice.
\end{theorem}

\begin{proof}
As usual, we prove this for $\lambda=1$ and appeal to the Blow-Up Lemma (\Cref{lem:blowup}) to generalize.

We reduce from $\intfind_N$, with $N=\binom{n}{2}$. Let Alice and Bob treat their inputs as $(x_{ij})_{1 \le i < j \le n}$ and $(y_{ij})_{1 \le i < j \le n}$ in some canonical way. Alice (resp.~Bob) converts their input into an edge stream consisting of pairs $(i,j)$ such that $i<j$ and $x_{ij}=0$ (resp.~$y_{ij}=0$). The concatenation of these streams defines the multigraph $\hat{G}$ given to the coloring algorithm. Let $(h,k)$ be the unique pair such that $x_{hk} = y_{hk} = 1$. Note that the underlying simple graph $G$ is $K_n$ minus the edge $\{h,k\}$. Therefore, $\degen = n-2$ and so, in a proper $(n-1)$-coloring of $\hat{G}$, there must be a repeated color and this can only happen at vertices $h$ and $k$.

Thus, a $p$-pass $(\degen+1)$-coloring algorithm using $s$ bits of space leads to a protocol for $\intfind_N$ using $(2p-1)s$ bits of communication. Therefore, $s = \Omega(N/p) = \Omega(n^2/p)$.
\end{proof}

\ifthenelse{\equal{\acmconf}{1}}
{

\subsection{Query Complexity Lower Bounds} \label{sec:query-lb}

In \Cref{sec:query-ub}, we gave a sublinear algorithm for 
producing a $(\degen+o(\degen))$-coloring in the general graph query model. 
We now prove lower bounds with the message that this can not be 
improved all the way to a $(\degen+1)$-coloring: that would preclude sublinear complexity. We also
give many generalizations of our lower bounds, similar in spirit to the
streaming lower bounds. Due to space limitation, we defer the
proofs of the theorems in this section to \Cref{app:query-lb-app}.
\footnote{For reader's convenience, we have duplicated \Cref{sec:query-lb}
in \Cref{app:query-lb-app} on its entirety; including the missing proofs of course.}

We first prove a lower bound for the $\gdct$ problem (see \Cref{def:graph-dist}).
This lower bound then serves as a building block for the coloring 
problem lower bound, as demonstrated in the previous section.
Due to space restrictions, we defer the proof of this lemma to
~\Cref{proof:query-lb-distinguish}.
\begin{lemma}
\label{lem:query-lb-distinguish-short}
Any randomized query algorithm that solves the \gdct($n,p,\lambda$) 
problem with success probability at least $2/3$, 
requires $\Omega(n^2/\lambda^2)$ many queries.
\end{lemma}

As an immediate consequence of \Cref{lem:query-lb-distinguish-short},
we get the following query lower bounds.
\begin{theorem}
Given a graph $G$ on $n$ vertices, and an integer parameter $\lambda \geq 1$,
any randomized query algorithm that, with probability at least $2/3$,
\begin{enumerate}[label=(\roman*),itemsep = 1pt]
    \item produces a $(\degen+\lambda)$-coloring of $G$,
            needs to make $\Omega(n^2/\lambda^2)$ queries.
    \item  outputs an estimate $\hat{\degen}$ for 
            $\degen(G)$ such that $|\hat{\degen} - \degen(G) | \leq  \lambda$,
            needs to make $\Omega(n^2/\lambda^2)$ queries.
\end{enumerate}
\end{theorem}
This $(\degen+\lambda)$-coloring lower bound, however, fails 
if $\degen$ is assumed to be known {\em a priori}.
So, we next address the question of whether knowing $\degen$ in advance
helps in producing a $(\degen+1)$-coloring in the query model. 
We show that the story here is same as that of the streaming model (see discussions
before \Cref{thm:streamcoloring-lb}); that is, any algorithm 
that receives a graph $G$ and $\degen(G)$ as input,
and produces a $(\degen+c)$-coloring of $G$ as output,
for any constant $c$, must make $\Omega(n^2)$ many queries. 
In fact, we prove a more general lower bound, as captured
by the following theorem. Again,
we defer the proof of this theorem to \Cref{proof:query-colorlb-lambda}.
\begin{theorem} \label{thm:query-colorlb-lambda-short}
  Suppose that a randomized algorithm in the general graph 
  query model reads an $n$-vertex input graph $G$ and 
  its degeneracy $\degen(G)$, and 
  with probability at least $2/3$, 
  produces a proper coloring of $G$ using at most $\degen(G)+\lambda$ colors. 
  Then the algorithm must make $\Omega(n^2/\lambda^2)$ queries.
\end{theorem}

}
{
\ifthenelse{\equal{\acmconf}{1}}
{
\section{Missing Proofs from Section~\ref{sec:query-lb}} \label{app:query-lb-app}
}
{
\subsection{Query Complexity Lower Bounds} \label{sec:query-lb}
}

We now turn to the general graph query model~\cite{Goldreich-proptest-book}.
Recall that our algorithm from \Cref{sec:query-ub} produces a 
$(\degen+o(\degen))$-coloring while making at most $\tO(n^{3/2})$ queries,
without needing to know $\degen$ in advance. Here, we shall prove that the
number of colors cannot be improved to $\degen+1$: that would preclude
sublinear complexity. In fact, we prove more general results,
similar in spirit to the streaming lower bounds from the previous section.
For these lower bounds, we use another family of gadget graphs.

\begin{definition} \label{def:query-gadget}
Given a large integer $p$ (a size parameter), the gadgets for that size are 
$(2p+1)$-vertex graphs on vertex set $A \uplus B$, where 
$A = \{a_1, \ldots, a_{p+1}\}$ and $B = \{b_1, \ldots, b_p\}$. Let
$H$ be the graph consisting of a clique on $A$ and a clique on $B$,
with no edges between $A$ and $B$. For $1 \le i < j \le p$, let
$H_{ij}$ be a graph on the same vertex set obtained by slightly modifying
$H$ as follows (see \Cref{fig:3}):
  \begin{equation} \label{eq:hij-def}
    E(H_{ij}) = E(H) \setminus \big\{\, \{a_i, a_j\},\{b_i, b_j\} \,\big\}
    \cup \big\{\, \{a_i, b_j\},\{a_j, b_i\} \,\big\} \,.
  \end{equation}
\end{definition}
Notice that the vertex $a_{p+1}$ is not touched by any of these modifications.
The relevant properties of these gadget graphs are as follows.
\begin{lemma} \label{lem:hij-degen}
  For all $1 \le i < j \le p$, $\degen(H_{ij}) = p-1$, whereas the chromatic number $\chi(H) = p+1$.
\end{lemma}
\begin{proof}
  The claim about $\chi(H)$ is immediate.

  Consider a particular graph $H_{ij}$. The subgraph induced by $A\setminus \{a_i\}$ is a $p$-clique, so $\degen(H_{ij}) \ge p-1$.

  Now consider the following ordering $\before$ for $H_{ij}$: $B\before a_i\before A\setminus \{a_i\}$, where the order within each set is arbitrary. For each $v\in B$, $\odeg_{H_{ij},\before}(v) \le \deg(v) = p-1$. For each $v \in A\setminus \{a_i\}$, $\odeg_{H_{ij},\before}(v) \le |A \setminus \{a_i\}| - 1 = p-1$. Finally, $a_i$ has exactly $p-1$ neighbors in $A\setminus \{a_i\}$ (by construction, $a_j$ is not a neighbor), so $\odeg_{H_{ij},\before}(a_i) = p-1$. By \Cref{lem:degen-odeg}, it follows that $\degen(H_{ij}) \le p-1$.
\end{proof}

Our proofs will use these gadget graphs in reductions from a pair of basic
problems in decision tree complexity. Consider inputs that are vectors in
$\b^N$: let $\bzero$ denote the all-zero vector 
and, for $i \in [N]$, let $\be_i$ denote the vector 
whose $i$th entry is $1$ while all other entries are $0$. 
Let $\uor_N$ and $\nih_N$ denote the following partial functions on $\b^N$:
\[
  \uor_N(\bx) = \begin{cases}
    0 \,, & \text{if } \bx = \bzero \,, \\
    1 \,, & \text{if } \bx = \be_i \,, \text{ for } i \in [N] \,, \\
    \star \,, & \text{otherwise;}
  \end{cases}
  \qquad
  \nih_N(\bx) = \begin{cases}
    i \,, & \text{if } \bx = \be_i \,, \text{ for } i \in [N] \,, \\
    \star \,, & \text{otherwise.}
  \end{cases}
\]
Informally, these problems capture, respectively, the tasks of 
(a)~determining whether there is a needle in a haystack under the
promise that there is at most one needle, and (b)~finding a needle
in a haystack under the promise that there is exactly one needle.
Intuitively, solving either of these problems with high accuracy 
should require searching almost the entire haystack. Formally,
let $\Rdt_\delta(f)$ denote the $\delta$-error randomized query complexity (a.k.a.~decision tree complexity) of $f$.
Elementary considerations of decision tree complexity lead to the bounds below (for a thorough discussion, including formal definitions, we refer the reader to the survey by Buhrman and de Wolf~\cite{BuhrmanW02}).

\begin{fact} \label{fact:uor-lb}
  For all $\delta \in (0, \frac12)$, we have $\Rdt_\delta(\uor_N)
  \ge (1-2\delta)N$ and $\Rdt_\delta(\nih_N) \ge (1-\delta)N-1$. \qed
\end{fact}

With this setup, we turn to lower bounds of the first flavor.
\begin{lemma} \label{lem:query-lb-distinguish}
  Solving $\gdct(n,p,\lambda)$ in the general graph query model requires $\Omega(n^2/\lambda^2)$ queries.
  
  More precisely, there is a constant $c > 0$ such that for every integer $\lambda \ge 1$ and every sufficiently
  large integer $p$, there is a setting $n = n(p,\lambda)$ for which every randomized query algorithm
  for $\gdct(n,p,\lambda)$ requires at least $c n^2/\lambda^2$ queries in the worst case.
\end{lemma}
\begin{proof}
We reduce from $\uor_N$, where $N = \binom{p}{2}$, using the following plan.
Put $n = (2p+1)\lambda$. Let $\cC$ be a query algorithm for $\gdct(n,p,\lambda)$.
Based on $\cC$, we shall design a $\frac13$-error algorithm $\cA$ for $\uor_N$ that
makes at most as many queries as $\cC$. By \Cref{fact:uor-lb}, this number of
queries must be at least $N/3 = \Omega(p^2) = \Omega(n^2/\lambda^2)$.

As usual, we detail our reduction for $\lambda=1$; the
Blow-up Lemma (\Cref{lem:blowup}) then handles general $\lambda$. By \Cref{lem:hij-degen},
$H \in \cG_1$ whereas each $H_{ij} \in \cG_2$ (cf.~\Cref{def:graph-dist}, taking $q = p$).

We now design $\cA$. Let $\bx \in \b^{N}$ be the input to $\cA$. Using a canonical bijection, let
us index the bits of $\bx$ as $x_{ij}$, where $1 \le i < j \le p$. Algorithm $\cA$ simulates
$\cC$ and outputs $1$ iff $\cC$ decides that its input lies in $\cG_2$. Since $\cC$
makes queries to a graph, we shall design an oracle for $\cC$ whose answers,
based on query answers for input $\bx$ to $\cA$, will implicitly define a graph on vertex set 
$V := A \uplus B$, as in \Cref{def:query-gadget}.
The oracle answers queries as follows.
  \begin{itemize}[itemsep=1pt]
    \item For $i,j \in [p]$, it answers $\Pair(a_i, a_j)$ and $\Pair(b_i, b_j)$ with $1-x_{ij}$.
    \item For $i,j \in [p]$, it answers $\Pair(a_i, b_j)$ and $\Pair(a_j, b_i)$ with $x_{ij}$.
    \item For $i \in [p]$, it answers $\Pair(a_{p+1}, a_i)$ with $1$ and $\Pair(a_{p+1}, b_i)$ with $0$.
    \item For $i\in [p]$ and $d\in [p-1]$, it answers $\Nbr(a_i, d)$ with $a_j$ if $x_{ij} = 0$ and $b_j$ if $x_{ij} = 1$, where $j=d$ if $d<i$, and $j=d+1$ otherwise.
    \item For $i, d\in [p]$, it answers $\Nbr(a_i, p)$ with $a_{p+1}$ and $\Nbr(a_{p+1}, d)$ with $a_d$.
    \item For $i\in [p]$ and $d\in [p-1]$, it answers $\Nbr(b_i, d)$ with $b_j$ if $x_{ij} = 0$ and $a_j$ if $x_{ij} = 1$, where $j=d$ if $d<i$, and $j=d+1$ otherwise.
    \item For all other combinations of $v \in V$ and $d \in \NN$, it answers $\Nbr(v,d) = \bot$.
  \end{itemize}
By inspection, we see that the graph defined by this oracle is $H$ if $\bx = \bzero$ and is $H_{ij}$ if $\bx = \be_{ij}$. Furthermore, the oracle answers each query by making at most one query to the input $\bx$. It follows that $\cA$ makes at most as many queries as $\cC$ and decides $\uor_N$ with error at most $\frac13$. This completes the proof for $\lambda=1$.

To handle $\lambda > 1$, we modify the oracle in the natural way so that the implicitly defined graph is $H^\lambda$ when $\bx = \bzero$ and $H_{ij}^\lambda$ when $\bx = \be_{ij}$. We omit the details, which are routine.
\end{proof}

As an immediate consequence of \Cref{lem:query-lb-distinguish},
we get the following query lower bounds.
\begin{theorem}
  Given query access to an $n$-vertex graph $G$, succeeding with probability
  at least $2/3$ at either of the following tasks requires $\Omega(n^2/\lambda^2)$
  queries, where $\lambda \ge 1$ is an integer parameter:
  \begin{enumerate}[label=(\roman*), itemsep=1pt]
    \item produce a proper $(\degen+\lambda)$-coloring of $G$;
    \item produce an estimate $\hat{\degen}$ such that $|\hat{\degen} - \degen| \le  \lambda$.
    \qed
  \end{enumerate}
\end{theorem}

We now prove a lower bound of the second flavor, where the algorithm
knows $\degen$ in advance.

\begin{theorem} \label{thm:query-colorlb-lambda}
  Given an integer $\degen$ and query access to an $n$-vertex graph $G$ with $\degen(G) = \degen$,
  an algorithm that, with probability $\frac23$, produces a proper $(\degen+\lambda)$-coloring of $G$
  must make $\Omega(n^2/\lambda^2)$ queries.
\end{theorem}

\begin{proof}
  We focus on the case $\lambda=1$; the general case is handled by the Blow-up Lemma, as usual.
  
  Let $\cC$ be an algorithm for the coloring problem. We design
  an algorithm $\cA$ for $\nih_N$, where $N = \binom{p}{2}$, using
  the same reduction as in \Cref{lem:query-lb-distinguish}, changing the
  post-processing logic as follows: $\cA$ outputs $(i,j)$ as its
  answer to $\nih_N(\bx)$, where $1 \le i < j \le p$ is such that
  $a_i$ and $a_j$ are colored the same by $\cC$.
  
  To prove the correctness of this reduction, note that when $\bx = \be_{ij}$,
  the graph defined by the simulated oracle is $H_{ij}$ and $\degen(H_{ij})
  = p-1$ (\Cref{lem:hij-degen}). Suppose that $\cC$ is successful, which 
  happens with probability at least $\frac23$. Then $\cC$ properly $p$-colors
  $H_{ij}$. Recall that $V(H_{ij}) = A \uplus B$, where $|A| = p+1$; there must
  therefore be a color repetition within $A$. The only two non-adjacent vertices
  inside $A$ are $a_i$ and $a_j$, so $\cA$ correctly answers $(i,j)$. By
  \Cref{fact:uor-lb}, $\cA$ must make $\Omega(N) = \Omega(p^2)$ queries.
\end{proof}


}

\subsection{A Combinatorial Lower Bound} 
\label{sec:impossibility}

Finally, we explore a connection between degeneracy based coloring and the {\em list coloring problem}.
In the latter problem, each vertex has a list of colors and
the goal is to find a corresponding list coloring---i.e., a proper coloring of the graph where
each vertex receives a color from its list---or to report that none exists.
Assadi \etal~\cite{AssadiCK19} proved a beautiful {\em Palette Sparsification Theorem},
a purely graph-theoretic result that connects the $(\Delta+1)$-coloring problem
to the list coloring problem.

Define a graph $G$ to be {\em $[\ell,r]_\delta$-randomly list colorable} (briefly, $[\ell,r]_\delta$-RLC) if choosing $r$ random colors per vertex, independently and uniformly without replacement from the palette $[\ell]$, permits a list coloring with probability at least $1-\delta$ using these chosen lists.\footnote{When $r \ge l$, this procedure simply produces the list $[\ell]$ for every vertex.} Their theorem can be paraphrased as follows.

\begin{fact}[Assadi et al.\cite{AssadiCK19}, Theorem 1]
  There exists a constant $c$ such that every $n$-vertex graph $G$ is $[\Delta(G)+1,\, c\log n]_{1/n}$-RLC. \qed
\end{fact}

Indeed, this theorem is the basis of the various coloring results in
their work. Let us outline how things work in the streaming model, focusing on the space usage.
Given an input graph $G$ that is promised to be $[\ell,r]_{1/3}$-RLC, for some
parameters $\ell,r$ that may depend on $G$, we sample $r$ random colors 
from $[\ell]$ for each vertex before reading the input. Chernoff bounds
imply that the {\em conflict graph}---the subgraph of $G$ consisting
only of edges between vertices whose color lists intersect---is of size
$O(|E(G)| r^2/\ell)$, w.h.p.. Using $|E(G)| \le n\Delta/2$, taking $\ell = \Delta+1$ and $r = O(\log n)$ bounds this size by $\tO(n)$, so a semi-streaming space bound suffices to collect the entire conflict graph.
(For full details, see Lemma 4.1 in \cite{AssadiCK19}.) Finding a list coloring of the conflict graph (which exists with probability at least $2/3$) yields an $\ell$-coloring of $G$.

For a similar technique to work in our setting, we would want $\ell \approx \degen$. Recalling that $|E(G)| \le n\degen$, for the space usage to be $\tO(n)$, we need $r=O(\polylog n)$. This raises the following combinatorial question: what is the smallest $\lambda$ for which we can guarantee that every graph is $[\degen+\lambda, O(\polylog n)]_{1/3}$-RLC?

By the discussion above, our streaming lower bound in \Cref{thm:streamcoloring-lb} already tells us that such a result is not possible with $\lambda = O(\degen^{\frac12-\gamma})$. Our final result (\Cref{thm:comb-lb} below) proves that we can say much more.

\medskip Let $J_{n,t}$ denote the graph $K_t + \overline{K}_{n-t}$, i.e., the graph join of a $t$-clique and an $(n-t)$-sized independent set. More explicitly,
\begin{equation} \label{eq:hnt-def}
  J_{n,t} = (A \uplus B, E) \,, \quad\text{where }
  |A| = t,\, |B| = n-t,\, E = \{\{u,v\}:\, u \in A, v \in A \cup B, u \ne v\} \,.
\end{equation}


\begin{lemma}
\label{lem:comb-lb}
  For integers $0 < r \le t < n$, if $J_{n,t}$ is $[\degen+\degen/r,\, r]_{\delta}$-RLC, then $\delta \ge 1 - r^n/(r+1)^{n-t}$.
\end{lemma}

\begin{proof}
  Take a graph $J_{n,t}$ with vertices partitioned into $A$ and $B$ as in \cref{eq:hnt-def}. An ordering with $B \lhd A$ shows that $\degen = \degen(J_{n,t}) = t$. We claim that for every choice of colors lists for vertices in $A$, taken from the palette $[t+t/r]$, the probability that the chosen lists for $B$ permit a proper list coloring is at most $p := r^n/(r+1)^{n-t}$. This will prove that $\delta \ge 1-p$.
  
  To prove the claim, consider a particular choice of lists for $A$. Fix a partial coloring $\psi$ of $A$ consistent with these lists. If $\psi$ is not proper, there is nothing to prove. Otherwise, since $A$ induces a clique, $\psi$ must assign $t$ distinct colors to $A$. In order for a choice of lists for $B$ to permit a proper extension of $\psi$ to the entire graph, every vertex of $B$ must sample a color from the remaining $t/r$ colors in the palette. Since $r$ colors are chosen per vertex, this event has probability at most
  \[
    \left(r \cdot \frac{t/r}{t+t/r}\right)^{|B|} = \left(\frac{r}{r+1}\right)^{n-t} \,.
  \]
  The claimed upper bound on $p$ now follows by a union bound over the $r^t$ possible partial colorings $\psi$.
\end{proof}

This easily leads to our combinatorial lower bound, given below. In reading the theorem statement, note that the restriction on edge density {\em strengthens} the theorem.
\begin{theorem} \label{thm:comb-lb}
  Let $n$ be sufficiently large and let $m$ be such that $n \le m \le n^2/\log^2 n$.
  If every $n$-vertex graph $G$ with $\Theta(m)$ edges is 
  $[\degen(G)+\lambda,\, c\log n]_{1/3}$-RLC for some parameter $\lambda$ and some
  constant $c$, then we must have $\lambda > \degen(G)/(c\log n)$.
\end{theorem}
\begin{proof}
  Suppose not. Put $t = \ceil{m/n}$, $r = c\log n$, and
  consider the graph $J_{n,t}$ defined in \cref{eq:hnt-def}.
  By the bounds on $m$, $|E(J_{n,t})| = t(t-1)/2 + t(n-t) = \Theta(nt) = \Theta(m)$.
  Put $\degen := \degen(J_{n,t})$. By assumption, $J_{n,t}$ is
  $[\degen + \degen/r,\, r]$-RLC, so \Cref{lem:comb-lb} implies that
  \[
    \frac23 \le \frac{r^n}{(r+1)^{n-t}}
    = \left(1 - \frac1{r+1}\right)^n (r+1)^t
    \le \exp\left( - \frac{n}{r+1} + t \ln(r+1) \right) \,.
  \]
  Since $t = O(n/\log^2 n)$ and $r = c\log n$, this is a contradiction for sufficiently large $n$. \end{proof}

We remark that the above result rules out the possibility of using a palette sparsification theorem along the lines of Assadi et al.~\cite{AssadiCK19} to obtain a semi-streaming coloring algorithm that uses fewer colors than \Cref{alg:streaming} (with the setting $\eps = 1/\log n$).

More generally, suppose we were willing to tolerate a weaker notion of palette sparsification by sampling $O(\log^d n)$ colors per vertex, for some $d\ge 1$: this would increase the space complexity of an algorithm based on such sparsification by a $\polylog n$ factor. By \Cref{lem:comb-lb}, arguing as in \Cref{thm:comb-lb}, we would need to spend at least $\degen + \degen/\Theta(\log^d n)$ colors. This is no better than the number of colors obtained using \Cref{alg:streaming} with the setting $\eps = 1/\log^d n$, which still maintains semi-streaming space. In fact, palette sparsification does not immediately guarantee a post-processing runtime that is better than exponential, because we need to color the conflict graph in post-processing. Meanwhile, recall that \Cref{alg:streaming} has $\tO(n)$ post-processing time via a straightforward greedy algorithm. Furthermore, since there exist ``hard'' graphs $J_{n,t}$ at all edge densities from $\Theta(n)$ to $\Theta(n^2/\log^2 n)$, we cannot even hope for a semi-streaming palette-sparsification-based algorithm that might work only for sparse graphs or only for dense graphs.

\section*{Acknowledgement}

We gratefully acknowledge several helpful discussions we have had with Sepehr Assadi (especially those that called to our attention a nuance with the Congested Clique algorithm) and Deeparnab Chakrabarty.

\ifthenelse{\equal{\acmconf}{1}}
{
    \bibliographystyle{ACM-Reference-Format}
}
{
    \bibliographystyle{alpha}
}

\newcommand{\etalchar}[1]{$^{#1}$}

\newpage
\appendix
\ifthenelse{\equal{\acmconf}{1}}
{

}

\end{document}